\documentclass{article}

\PassOptionsToPackage{numbers}{natbib}
\usepackage[main,final]{neurips_2026}

\usepackage[utf8]{inputenc}
\usepackage[T1]{fontenc}
\usepackage{hyperref}
\usepackage{url}
\usepackage{graphicx}
\usepackage{wrapfig}
\usepackage{enumitem}
\usepackage{nameref}
\usepackage{booktabs}
\usepackage[hypcap=false]{caption}
\usepackage{amsmath,amssymb,amsthm}
\theoremstyle{plain}
\newtheorem{proposition}{Proposition}

\theoremstyle{definition}

\usepackage{nicefrac}
\usepackage{microtype}
\usepackage{xcolor}
\usepackage{empheq}
\definecolor{lightteal}{RGB}{220,245,245}
\newcommand*{\eqbox}[1]{\colorbox{lightteal}{#1}}

\title{PACE: Geometry-Aware Bridge Transport for Single-Cell Trajectory Inference}

\author{Chenglei Yu$^{1,2}$\thanks{Equal contribution. \quad \textdagger Corresponding author.}, Chuanrui Wang$^{2}$\footnotemark[1], Bangyan Liao$^{1,2}$ \& Tailin Wu$^{2}\textsuperscript{\textdagger}$ \\
$^{1}$Zhejiang University\\
$^{2}$Department of Artificial Intelligence, School of Engineering, Westlake University\\
\texttt{\{yuchenglei, wangchuanrui, wutailin\}@westlake.edu.cn} \\
}

\begin{document}

\newcommand{\model}{PACE}

\maketitle

\begin{abstract}
Single-cell trajectory inference from destructive time-course snapshots is fundamentally ill-posed: neither cross-time cell correspondences nor the continuous paths between snapshots are observed, so the observed snapshot distributions alone do not uniquely determine the underlying dynamics. 
Existing optimal transport and flow-based methods typically couple cells by Euclidean proximity at observed clock times, which can misalign trajectories when development is asynchronous and cells sampled at the same experimental time occupy different latent pseudotime stages. 
We propose PACE, a trajectory inference framework that selects geometry-consistent continuous transport dynamics from destructive time-course snapshots through three coupled components.
First, PACE constructs a state- and time-dependent anisotropic Riemannian metric that preserves low cost along locally supported tangent directions while penalizing normal velocity components.
Second, it alternates between refining cross-time couplings under the induced path-action cost and fitting endpoint-preserving neural bridges between adjacent snapshots.
Third, it distills the learned bridge dynamics into a global continuous-time velocity field over cellular states.
Across seven controlled and biological datasets covering nine held-out reconstruction experiments, PACE achieves the strongest overall reconstruction performance, reducing MMD, $\mathcal{W}_1$, and $\mathcal{W}_2$ by 23.7\% on average relative to the strongest competing baseline.
PACE also improves RNA-velocity alignment by 15.4\% on an embryoid body differentiation benchmark, without requiring explicit cell pairing, lineage tracing, or RNA velocity supervision during training.
Code is available at \url{https://github.com/AI4Science-WestlakeU/PACE}.
\end{abstract}

\section{Introduction}
\label{sec:intro}

Understanding how cells move from one state to another is a central problem in single-cell biology~\citep{farrell2018single,wagner2018single,schiebinger2019optimal,tanay2017scaling,griffiths2018using}. Processes such as development, differentiation, immune activation, tumor evolution, and cellular reprogramming are not simply collections of discrete cell types, but continuous population-level transitions through high-dimensional molecular state space~\citep{trapnell2014dynamics,haghverdi2016diffusion,setty2019characterization,cacchiarelli2018aligning}. The key questions are therefore dynamical, including which early states commit to particular fates, which intermediate states are transient but decisive, and which regulatory programs drive these transitions. Trajectory inference~\citep{lavenant2021towards,hashimoto2016learning} aims to answer these questions by reconstructing continuous cell-state evolution from time-course single-cell observations.

Transport- and flow-based trajectory methods provide a natural framework for modeling such dynamics~\citep{schiebinger2019optimal, peyre2019computational,lavenant2021towards}. Since time-course single-cell data provide population-level snapshots rather than paired observations of the same cells~\citep{schiebinger2019optimal, weinreb2018fundamental}, these methods must infer how marginal distributions are coupled across time. Most existing formulations operate in the Euclidean representation space, where couplings are estimated from Euclidean endpoint costs and continuous paths or velocity fields are regularized by Euclidean kinetic energy, action, or smoothness~\citep{neklyudov2023action,tong2023conditional,tong2020trajectorynet,huguet2022manifold}. This Euclidean geometry is convenient, but it is not necessarily aligned with biological progress. Two cells can be close in expression space while lying at different developmental stages or fate branches, whereas biologically plausible motion may follow curved and locally anisotropic directions along the cell-state manifold~\citep{moon2019visualizing,wolf2019paga}.

This mismatch makes trajectory inference from destructive snapshots fundamentally ambiguous~\citep{weinreb2018fundamental,tritschler2019concepts}: the cross-time coupling is not observed and is generally not identifiable from marginal distributions alone. A coupling may satisfy the marginal constraints while connecting cells across incompatible developmental programs, stalled states, or fate branches. For example, in mouse reprogramming~\citep{schiebinger2019optimal}, differences in growth rates between cell types can make marginal matching misleading, coupling apoptotic stromal cells to rapidly expanding iPSCs. Human iPSC reprogramming~\citep{liu2020reprogramming} shows the same issue at finer granularity, that cells collected on the same day can contain heterogeneous primed-like, naive-like, and trophectoderm-like intermediates, so clock-time adjacency alone does not determine which cells should be coupled. 
Auxiliary measurements such as RNA velocity, lineage tracing, or metabolic labeling~\citep{la2018rna,bergen2020generalizing,lange2022cellrank} can help, but require additional experiments and are unavailable in many datasets. Existing Euclidean or support-based regularizers can encourage short, smooth, or data-supported paths~\citep{saelens2019comparison}, but they do not directly encode which local directions of motion are developmentally admissible.

We propose PACE, a geometry-aware trajectory-inference framework for selecting snapshot-consistent continuous transport dynamics from unpaired time-course snapshots.
The intuition is that, although true cell identities and intermediate paths are unobserved, local spatiotemporal neighborhoods still provide a weak but useful prior on admissible directions, since motion along locally supported tangent directions is more plausible than motion orthogonal to the observed manifold structure.
PACE therefore replaces Euclidean transport costs with a Riemannian path action induced by a state- and time-dependent anisotropic metric, which assigns lower cost to locally supported tangent motion and higher cost to normal motion.
Specifically, PACE estimates local tangent \emph{directions} from spatiotemporal neighborhoods of observed cells and uses this metric to define an anisotropic bridge transport problem, where the cost of coupling two cells is the minimum Riemannian action of an endpoint-conditioned path.
PACE then alternates between refining OT couplings under this path-action cost and fitting endpoint-preserving neural bridges, allowing local geometry to influence both which cells are coupled and how they move between snapshots.
Finally, the learned bridge dynamics are distilled into a continuous-time population velocity field over cellular states.

Our contributions are threefold:
\begin{enumerate}
    \item We introduce a time- and state-dependent Riemannian metric based on local spatiotemporal tangent-space projections, providing a geometry-aware path-action cost for selecting couplings between adjacent destructive snapshots beyond Euclidean endpoint proximity.

    \item We develop an iterative bridge-transport procedure that alternates between Riemannian coupling refinement and endpoint-preserving neural bridge fitting, allowing cross-time correspondences and interpolant paths to be selected jointly under the same geometry-aware action.

\item 
We distill endpoint-conditioned bridges into a global continuous-time population velocity field and evaluate PACE on seven datasets covering nine held-out reconstruction experiments. 
PACE achieves the strongest overall reconstruction performance, reducing MMD, $\mathcal{W}_1$, and $\mathcal{W}_2$ by 23.7\% on average relative to the strongest competing baseline, improves RNA-velocity alignment on an embryoid body differentiation benchmark, and shows consistent gains across component ablations.
\end{enumerate}

\section{Related Work}
\label{sec:related}

\paragraph{Optimal transport and flow-based trajectory inference.}
Optimal transport provides a natural framework for coupling unpaired snapshot distributions in single-cell analysis \citep{schiebinger2019optimal,lavenant2021towards}. Recent neural extensions such as TrajectoryNet \citep{tong2020trajectorynet}, Wasserstein Lagrangian Flows \citep{neklyudov2023action}, and Conditional Flow Matching (CFM) \citep{tong2023conditional, lipman2023flow} learn continuous dynamics by regressing vector fields along interpolants between coupled endpoints. Schrodinger bridge methods \citep{debortoli2021diffusion,shi2023diffusion,chen2023deep} add stochasticity or entropic regularization to the transport problem, while Curly Flow Matching \citep{petrovic2026curly} extends the framework to non-gradient dynamics using approximate velocity information. These methods typically infer couplings from Euclidean proximity or OT distances in the observed space, which can misalign trajectories when cells at the same experimental time occupy different pseudotime stages.

\paragraph{Geometry-aware generative models.}
The manifold hypothesis has motivated data-dependent Riemannian metrics in ambient spaces~\citep{hauberg2012geometric,arvanitidis2017latent}, flow matching on known manifolds or general geometries~\citep{chen2024flow}, and manifold-aware OT flows for single-cell trajectories~\citep{huguet2022manifold}.
Metric Flow Matching (MFM)~\citep{kapusniak2024metric} is closest to our setting: it uses task-independent \emph{support-aware} metrics such as LAND~\citep{arvanitidis2017latent} and RBF to pull geodesics toward the data support.
PACE differs in both the source and use of geometry.
In asynchronous reprogramming, density support alone does not determine plausible cross-time transitions, since a high-density intermediate may connect to either progressed or refractory fates and Euclidean proximity can couple incompatible programs.
PACE instead builds a time- and state-dependent, \emph{direction-aware} metric from local spatiotemporal tangent subspaces estimated from destructive snapshots, penalizing motion orthogonal to plausible developmental directions.
The resulting metric action is used both to learn interpolant paths and to refine cross-time couplings, whereas MFM typically assumes the endpoint pairing is fixed before interpolant learning.

\section{Method}
\label{sec:method}

\begin{figure*}[t]
    \centering
    \includegraphics[width=\textwidth]{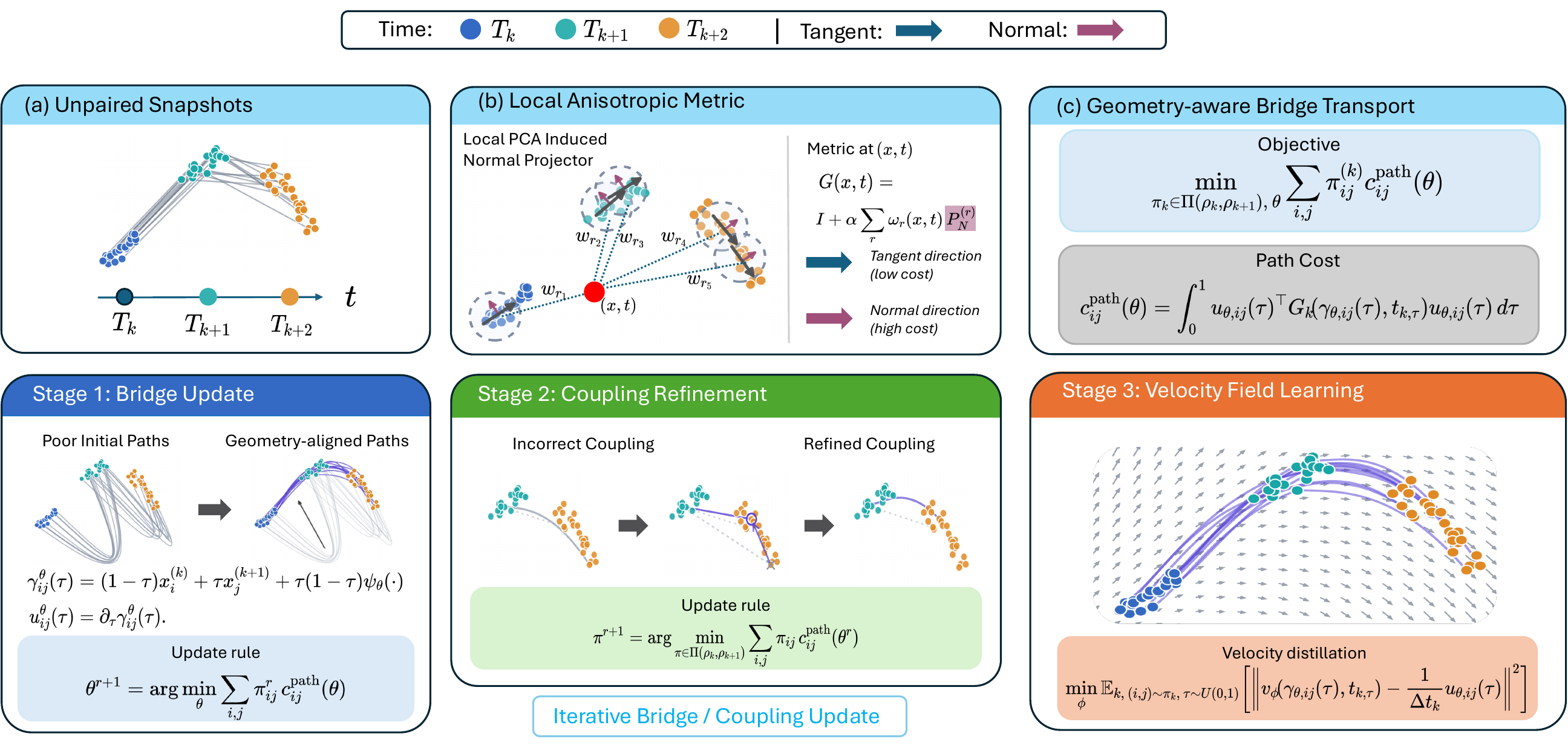}
    \caption{
    \textbf{Overview of PACE.}
    PACE uses local PCA to construct an anisotropic metric \(G_k(x,t)=I+\alpha C_N^{(k)}(x,t)\), trains endpoint-preserving neural bridges under the corresponding path-action cost, iteratively refines cross-time couplings, and distills the learned bridge dynamics into a global velocity field for trajectory inference from unpaired snapshots.
    }
    \label{fig:pace_overview}
\end{figure*}

\subsection{Problem formulation}
\label{sec:problem-formulation}

We observe single-cell point clouds at anchor times
\(\{t_k\}_{k=0}^K\), with \(t_0=0\),
\[
\mathcal{X}^{(k)}
=
\{x_i^{(k)}\}_{i=1}^{N_k}
\sim \hat\rho_k
\subset \mathbb{R}^d,
\qquad
k=0,1,\dots,K.
\]
Here \(d\) denotes the representation dimension.
Our goal is to infer trajectories and a velocity field transporting \(\hat\rho_k\) to \(\hat\rho_{k+1}\) between adjacent times.
Once learned, the velocity field can be integrated from \(x_{t_0}\sim\hat\rho_0\) to generate trajectories over the time course.

\begin{proposition}[Ill-posedness from snapshots alone]
\label{prop:illposed}
Given only the point clouds \(\{\mathcal{X}^{(k)}\}_{k=0}^K\), the reconstruction of cross-time couplings and intermediate trajectories is non-unique.
\end{proposition}

Indeed, for a single interval \([t_k,t_{k+1}]\), many couplings \(\pi_k\in\Pi(\hat\rho_k,\hat\rho_{k+1})\) share the same endpoint marginals. For any such coupling, each coupled endpoint pair can also be connected by infinitely many smooth paths. Thus, the observed snapshots determine neither a unique correspondence nor a unique intermediate trajectory. This observation motivates PACE to frame trajectory inference as a variational problem over couplings and paths, selected by a geometry-aware path action. The proof of Proposition~\ref{prop:illposed} is given in Appendix~\ref{app:illposed}.

\subsection{Geometry-aware bridge transport}
\label{sec:variational-principle}

To address the ill-posedness in Proposition~\ref{prop:illposed}, PACE formulates trajectory reconstruction as selecting both a cross-time coupling and a family of continuous paths, chosen by minimizing a geometry-aware action.
Given a time- and state-dependent metric tensor \(G_k(x,t)\) (constructed in \S\ref{sec:metric}), the cost of a path \(\gamma\), or 
"path action", between endpoints \((x,y)\) is
\begin{equation}
\label{eq:pace-path-action}
\mathcal{A}_{G_k}[\gamma]
=
\int_0^1
\dot\gamma(\tau)^\top
G_k\!\bigl(\gamma(\tau),\,t_{k,\tau}\bigr)
\dot\gamma(\tau)
\,d\tau,
\qquad
\gamma(0)=x,\;\gamma(1)=y,
\end{equation}
where \(\tau\in[0,1]\) is local interpolation time and the physical time is \(t_{k,\tau}=t_k+\tau(t_{k+1}-t_k)\).
The path action first defines the geometry-aware endpoint cost, and the ideal PACE bridge problem then transports mass using this induced cost:
\begin{empheq}[box=\eqbox]{equation}
\label{eq:ideal-pace-bridge}
\begin{aligned}
c_{G_k}(x,y)
&=
\inf_{\substack{\gamma(0)=x\\ \gamma(1)=y}}
\mathcal{A}_{G_k}[\gamma],&
\pi_k^\star
\in
\arg\min_{\pi_k\in\Pi(\hat\rho_k,\hat\rho_{k+1})}
\int
c_{G_k}(x,y)
\,d\pi_k(x,y).
\end{aligned}
\end{empheq}
Here \(c_{G_k}(x,y)\) is not a fixed Euclidean endpoint distance; it is the minimum action required to move from \(x\) to \(y\) under the metric \(G_k\).

This formulation makes explicit how PACE differs from standard OT.
If \(G_k(x,t)\equiv I\), then Eq.~\eqref{eq:pace-path-action} reduces to the Euclidean kinetic energy.
For any fixed endpoint pair \((x,y)\), the minimum-action path is the straight constant-speed interpolant, and the induced cost becomes
$
c_{G_k}(x,y)=\|y-x\|^2.
$
In this special case, Eq.~\eqref{eq:ideal-pace-bridge} reduces to standard quadratic-cost OT between \(\hat\rho_k\) and \(\hat\rho_{k+1}\).

PACE departs from this Euclidean case by using a state- and time-dependent anisotropic metric \(G_k(x,t)\).
The induced cost \(c_{G_k}(x,y)\) is no longer determined only by endpoint distance; it depends on the entire path and on how the \emph{path velocity} aligns with local developmental geometry.
As a result, the minimum-action cost generally has no closed-form solution and cannot be reduced to a fixed Euclidean endpoint cost.
The next section defines the spatiotemporal metric \(G_k(x,t)\), and \S~\ref{sec:alternating} describes how PACE approximates this 
problem with neural bridges and alternating coupling updates.

\paragraph{Finite-dimensional approximation.}
Problem~\eqref{eq:ideal-pace-bridge} is an ideal infinite-dimensional formulation.
PACE makes it tractable through two approximations (\S\ref{sec:alternating}):
(i) the path family is parameterized by neural bridges \(\gamma_\theta\) that satisfy endpoint constraints by construction;
(ii) the action integral is approximated on a finite time grid.
These yield a finite-dimensional alternating optimization over bridge parameters and correspondence variables.

\subsection{Time-dependent spatiotemporal tangent metric}
\label{sec:metric}

The metric should encode a simple prior: admissible motion should follow directions locally supported by the observed geometry, rather than cut across the cell-state manifold. Although local neighborhoods do not reveal the true direction of time, their dominant variation directions provide an undirected tangent approximation to the set of plausible state changes.
PACE therefore treats tangent motion as low cost and penalizes velocity components in the locally estimated normal subspace.
Because the local composition and geometry of snapshots can change across experimental time, the admissible subspaces are indexed by both state and time, yielding a time- and state-dependent metric \(G_k(x,t)\).
	        
        \paragraph{Local normal subspaces.}
        At each anchor point \(x_r\) observed at time \(t_r\), PACE estimates an anchor-wise local normal direction by finding its \(m_{\rm nn}\) nearest spatial neighbors within the same snapshot (including \(x_r\) itself), computing a Gaussian-kernel weighted covariance of the neighbor cloud, and extracting the minimum-variance principal direction \(n_r\).
        In two dimensions the normal projector is simply \(P_N^{(r)}=n_r n_r^\top\); in higher dimensions PACE adaptively selects the tangent-subspace dimension
        and builds the normal projector as the orthogonal complement (Appendix~\ref{app:adaptive-projectors}).
        
        \paragraph{Spatiotemporal metric construction.}
        For a query point \((x,t)\) in segment \(k\), PACE interpolates these anchor-wise local normal projectors across state and time using space-time Gaussian weights over all anchor points \(r\):
        \begin{equation}
            \label{eq:weights}
            \omega_r^{(k)}(x,t)
            \propto
            \exp\!\left(
            -\frac{\|x-x_r\|^2}{(h_x^{(k)})^2}
            -\frac{|t-t_r|^2}{(h_t^{(k)})^2}
            \right),
        \end{equation}
        with normalization chosen so that \(\sum_r\omega_r^{(k)}(x,t)=1\), where the bandwidths \(h_x^{(k)}\) and \(h_t^{(k)}\) are estimated adaptively for each segment (see Appendix~\ref{app:bandwidth}).
        These weights define an averaged normal projector field and the corresponding metric tensor
        \[
        C_N^{(k)}(x,t)=\sum_r \omega_r^{(k)}(x,t)P_N^{(r)}, \qquad G_k(x,t)=I+\alpha C_N^{(k)}(x,t),\qquad \alpha>0.
        \]
        The corresponding velocity cost is \(v^\top G_k(x,t)v = \|v\|^2 + \alpha\, v^\top C_N^{(k)}(x,t)\,v\). Because \(C_N^{(k)}\) is a positive-semidefinite average of local normal projectors, the second term increases the cost most strongly for velocities aligned with nearby estimated normal directions.

        \begin{proposition}[Normal-subspace penalizing property]
        \label{prop:metric}
        For every segment \(k\), query point \((x,t)\), and velocity \(v\in\mathbb{R}^d\),
        \begin{equation}
        \label{eq:metric-psd}
        v^\top G_k(x,t)v
        =
        \|v\|^2
        +
        \alpha v^\top C_N^{(k)}(x,t)v
        \ge
        \|v\|^2.
        \end{equation}
        Moreover, if \(C_N^{(k)}(x,t)=P_N(x,t)\) is an exact orthogonal projector onto the normal subspace \(T_{x,t}^{\perp}\), and \(v=v_T+v_N\) with \(v_T\in T_{x,t}\) and \(v_N\in T_{x,t}^{\perp}\), then
        \begin{equation}
        \label{eq:tangent-normal-cost}
        v^\top G_k(x,t)v
        =
        \|v_T\|^2
        +
        (1+\alpha)\|v_N\|^2.
        \end{equation}
        \end{proposition}
        
        Proposition~\ref{prop:metric} shows that the interpolated metric always preserves at least the Euclidean velocity cost, and that tangent motion retains its Euclidean cost while normal motion is penalized by a factor \(1+\alpha\) in the ideal projector case.
        Thus, the metric turns the qualitative principle of geometry-aligned cellular motion into an explicit action functional.

    \subsection{Finite-dimensional optimization by alternating bridge and OT coupling updates}
    \label{sec:alternating}

    PACE approximates the anisotropic bridge problem by alternating between an endpoint-conditioned bridge and a cross-time coupling. The coupling determines which endpoint pairs are used to train the bridge, while the bridge induces a path-action cost for updating the coupling. This creates a bootstrap mechanism in which more plausible couplings provide better endpoint supervision, and better bridges provide a geometry-aware approximation to the endpoint cost used for OT refinement. Since solving a separate minimum-action path problem for every candidate endpoint pair is infeasible under the state- and time-dependent metric $G_k(x,t)$, PACE amortizes path optimization with a shared endpoint-preserving neural bridge $\gamma_\theta(x,y,\tau)$. See Appendix~\ref{app:alternating-intuition} for a detailed discussion of this alternating procedure.

    \paragraph{Endpoint-preserving neural bridge.}
    For an endpoint pair \((x,y)\in\mathbb{R}^d\times\mathbb{R}^d\) and local time \(\tau\in[0,1]\), we define
    \begin{equation}
    \label{eq:bridge}
    \gamma_\theta(x,y,\tau)
    =
    (1-\tau)x+\tau y
    +
    \tau(1-\tau)\psi_\theta(x,y,\tau),
    \end{equation}
    where \(\psi_\theta:\mathbb{R}^d\times\mathbb{R}^d\times[0,1]\to\mathbb{R}^d\) is a neural network. The factor \(\tau(1-\tau)\) enforces \(\gamma_\theta(x,y,0)=x\) and \(\gamma_\theta(x,y,1)=y\), so the network only controls the interior deformation of the path~\citep{petrovic2026curly}. The bridge velocity is \(u_\theta(x,y,\tau)=\partial_\tau\gamma_\theta(x,y,\tau)\), computed by automatic differentiation. When \(\psi_\theta=0\), the bridge reduces to the straight Euclidean interpolant.

\paragraph{Bridge learning under fixed coupling.}
Assume that, for each adjacent snapshot pair, a coupling \(\pi_k\in\Pi(\hat\rho_k,\hat\rho_{k+1})\) is given. During the bridge update, endpoint pairs \((x_i^{(k)},x_j^{(k+1)})\sim\pi_k\) are sampled from the current coupling. We write \(t_{k,\tau}=t_k+\tau(t_{k+1}-t_k)\).

PACE trains the endpoint-preserving bridge by minimizing
\begin{equation}
\label{eq:total-loss}
\mathcal{L}_{\mathrm{bridge}}(\theta)
=
\lambda_{\mathrm{metric}}\mathcal{L}_{\mathrm{metric}}
+
\lambda_{\mathrm{reg}}\mathcal{L}_{\mathrm{reg}}.
\end{equation}
The main term $\mathcal{L}_{\mathrm{metric}}$ is the anisotropic bridge action induced by the spatiotemporal metric $G_k$.
An optional regularization term $\mathcal{L}_{\mathrm{reg}}$ penalizes large normal components and incoherent cross-segment velocities (Appendix~\ref{app:regularizers}).

The metric-action loss is
\begin{equation}
\label{eq:metric-loss}
\mathcal{L}_{\mathrm{metric}}
=
\mathbb{E}_{k,(x,y)\sim\pi_k}
\left[
\frac{1}{T}
\sum_{\ell=1}^{T}
u_\theta(x,y,\tau_\ell)^\top
G_k\!\left(
\gamma_\theta(x,y,\tau_\ell),
t_{k,\tau_\ell}
\right)
u_\theta(x,y,\tau_\ell)
\right].
\end{equation}
    
\paragraph{Coupling update under fixed bridge.}
After the bridge model has been updated under the current coupling, PACE recomputes the coupling using the full path action rather than Euclidean endpoint distance. For a candidate source-target pair \((x_i^{(k)},x_j^{(k+1)})\), the path-action cost is
\begin{equation}
\label{eq:path-cost}
c_{ij}^{\mathrm{path}}(\theta)
=
\frac{1}{M}
\sum_{m=1}^{M}
u_\theta(x_i^{(k)},x_j^{(k+1)},\tau_m)^\top
G_k\!\left(
\gamma_\theta(x_i^{(k)},x_j^{(k+1)},\tau_m),
t_{k,\tau_m}
\right)
u_\theta(x_i^{(k)},x_j^{(k+1)},\tau_m).
\end{equation}
The coupling is then updated by solving
\begin{equation}
\label{eq:ot-refined}
\pi_k
\leftarrow
\arg\min_{\pi_k\in\Pi(\hat\rho_k,\hat\rho_{k+1})}
\sum_{i,j}
c_{ij}^{\mathrm{path}}(\theta)\pi_{ij}^{(k)}.
\end{equation}
\paragraph{Alternating optimization.}
Given an initial coupling, PACE alternates between two blocks:
\begin{enumerate}
    \item \textbf{Bridge update.} Fix the current couplings \(\{\pi_k\}\), sample endpoint pairs \((x_i^{(k)},x_j^{(k+1)})\sim\pi_k\), and update \(\theta\) by taking gradient steps on Eq.~\eqref{eq:total-loss}.

    \item \textbf{Coupling update.} Fix the current bridge \(\gamma_\theta\), compute the path-action cost matrix in Eq.~\eqref{eq:path-cost}, and update each \(\pi_k\) by solving the OT problem in Eq.~\eqref{eq:ot-refined}.
\end{enumerate}
The bridge update learns low-action endpoint-conditioned paths under the current coupling, while the coupling update selects source-target transport using the learned path action.
In the idealized finite-dimensional setting where the bridge block is solved exactly for the metric-action objective, these two updates form a monotone block-coordinate descent scheme; Appendix~\ref{app:optimization-monotonicity} states this property and relates it to the stochastic neural implementation.
In practice, the coupling update is triggered periodically, for example every \(R_0\) epochs.
    
\subsection{Distilling endpoint-conditioned bridges into a global velocity field}
\label{sec:distillation}

The alternating optimization in \S\ref{sec:alternating} yields endpoint-conditioned bridges \(\gamma_\theta(x,y,\tau)\) and \(\tau\)-velocities \(u_\theta(x,y,\tau)=\partial_\tau\gamma_\theta(x,y,\tau)\). These bridges are pair-specific: to query a velocity at an arbitrary state \(z\) and time \(t\), one would need to know which endpoint pair \((x,y)\) and interpolation parameter \(\tau\) generated \(z\). PACE therefore distills the bridge dynamics into a global velocity field \(v_\phi(x,t)\) that can be evaluated without reference to a particular endpoint pair.

\paragraph{Distillation objective.}
For each segment \(k\), we sample \((x,y)\sim\pi_k\) and \(\tau\sim\mathrm{Uniform}[0,1]\), set \(x_{t_{k,\tau}}=\gamma_\theta(x,y,\tau)\) with \(t_{k,\tau}=t_k+\tau\Delta t_k\), and match the global velocity field:
\begin{equation}
\label{eq:distill-loss}
\mathcal{L}_{\mathrm{distill}}(\phi)
=
\mathbb{E}_{k,(x,y)\sim\pi_k}
\mathbb{E}_{\tau}
\left[
\left\|
v_\phi(x_{t_{k,\tau}},t_{k,\tau})
-
\frac{u_\theta(x,y,\tau)}{\Delta t_k}
\right\|^2
\right],
\end{equation}
where \(\Delta t_k=t_{k+1}-t_k\) converts \(\tau\)-velocity to physical-time velocity.
Once \(v_\phi\) is learned, continuous trajectories are generated by integrating
\begin{equation}
\label{eq:global-ode}
\frac{dx_t}{dt}=v_\phi(x_t,t),
\qquad x_0\sim\hat\rho_0.
\end{equation}

\begin{figure}[t]
    \centering
    \includegraphics[width=1.0\linewidth]{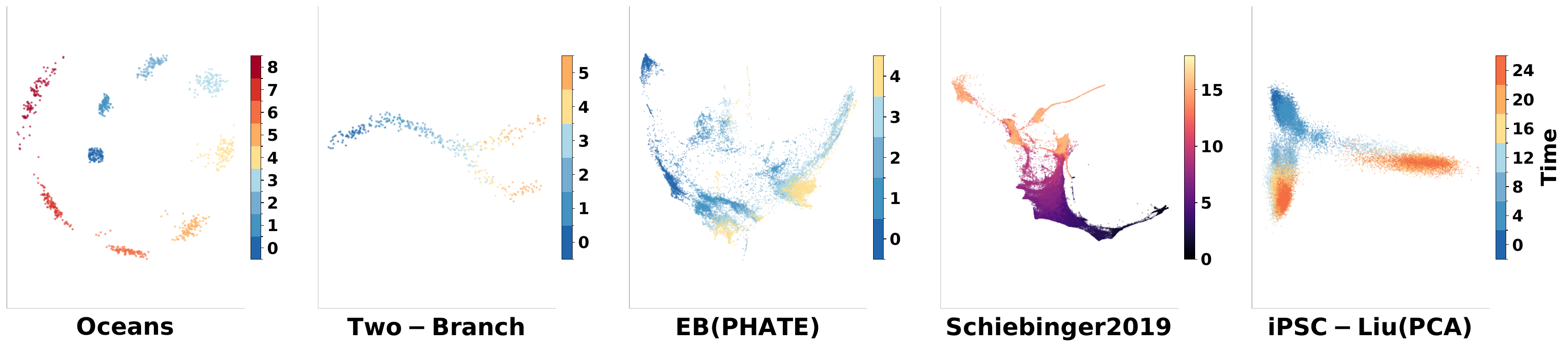}
    \caption{Overview of the 2D benchmark datasets. Points are colored by observed time for Ocean~\citep{petrovic2026curly}, Two-Branch, EB PHATE~\citep{moon2019visualizing}, Schiebinger2019~\citep{schiebinger2019optimal}, and iPSC-Liu~\citep{liu2020reprogramming}.}
    \label{fig:datasets-2d-overview}
\end{figure}
\begin{table}[hbtp]
\centering
\caption{Per-timepoint results on Ocean (2D)~\citep{petrovic2026curly}, holdout $t\in\{1,3,5,7\}$.}
\label{tab:ocean_2d_1_3_5_7}
\small
\setlength{\tabcolsep}{4pt}
\resizebox{1\linewidth}{!}{
\begin{tabular}{lcccccccccccc}
\toprule
Method & \multicolumn{4}{c}{MMD $\downarrow$} & \multicolumn{4}{c}{$\mathcal{W}_1\downarrow$} & \multicolumn{4}{c}{$\mathcal{W}_2\downarrow$} \\
\cmidrule(lr){2-5} \cmidrule(lr){6-9} \cmidrule(lr){10-13}
 & $t=1$ & $t=3$ & $t=5$ & $t=7$ & $t=1$ & $t=3$ & $t=5$ & $t=7$ & $t=1$ & $t=3$ & $t=5$ & $t=7$ \\
\midrule
Action Matching & \underline{0.7429} & 1.0662 & 0.9465 & 1.1289 & 0.3549 & 0.4823 & 0.2383 & 0.5073 & 0.3662 & 0.4846 & 0.2425 & 0.5143 \\
Aligned CFM & 0.7989 & 0.8754 & 1.0073 & 0.8272 & 0.0976 & 0.1644 & 0.3024 & 0.2245 & 0.1017 & 0.1673 & 0.3046 & 0.2262 \\
CURLY & 1.0545 & 0.8373 & 0.5869 & 0.8663 & 0.3766 & 0.1521 & 0.1077 & 0.2638 & 0.3771 & 0.1537 & 0.1102 & 0.2646 \\
DMSB & 1.1260 & \underline{0.6089} & \underline{0.2542} & \underline{0.4962} & 0.3060 & \underline{0.1097} & \underline{0.0435} & \underline{0.1097} & 0.3065 & \underline{0.1150} & \underline{0.0473} & \underline{0.1142} \\
MFM & 0.8482 & 0.7859 & 0.6928 & 0.5823 & \underline{0.0960} & 0.1433 & 0.1359 & 0.1279 & \underline{0.0991} & 0.1452 & 0.1385 & 0.1302 \\
OT-CFM & 0.9410 & 0.8169 & 0.6827 & 0.6305 & 0.1227 & 0.1428 & 0.1289 & 0.1321 & 0.1255 & 0.1459 & 0.1308 & 0.1334 \\
\midrule
\textbf{\model (ours)} & \textbf{0.4504} & \textbf{0.2588} & \textbf{0.0735} & \textbf{0.2779} & \textbf{0.0399} & \textbf{0.0398} & \textbf{0.0268} & \textbf{0.0505} & \textbf{0.0440} & \textbf{0.0434} & \textbf{0.0362} & \textbf{0.0534} \\
\bottomrule
\end{tabular}
}
\end{table}

\section{Experiment}
\label{sec:experiment}
The experiments evaluate PACE through five questions. First, on controlled 2D trajectories, can PACE recover smooth and branching geometry without paired identities or velocity supervision (Figure~\ref{fig:datasets-2d-overview}; Tables~\ref{tab:ocean_2d_1_3_5_7} and~\ref{tab:branch_toy_dim2})? Second, when reference velocities are available only for evaluation, do the learned dynamics align with held-out velocity fields on Ocean and EB PHATE (Figure~\ref{fig:velocity-diagnostics}; Appendix Tables~\ref{tab:ocean_holdout_velocity_metrics} and~\ref{tab:ebphate_holdout_velocity_metrics})? Third, on biological time courses, does PACE improve held-out reconstruction for single-cell differentiation and reprogramming from destructive snapshots (Table~\ref{tab:single_cell_2d_time_average})? Fourth, does PACE remain effective as the representation dimension increases (Tables~\ref{tab:ipsc_liu_dim10_50_mean})? Finally, which components of PACE account for the observed gains (Figure~\ref{fig:pace-ablation-per-tp})?
\subsection{Experimental setup}

\textbf{Evaluation protocol. }
All experiments use the same held-out snapshot protocol: methods observe only the training time points and are evaluated on the held-out time points listed in each table caption. We report distributional reconstruction quality using maximum mean discrepancy (MMD), 1-Wasserstein distance ($\mathcal{W}_1$), and 2-Wasserstein distance ($\mathcal{W}_2$); lower is better, bold denotes the best result, and underline denotes the second best.

\textbf{Datasets. }
Our benchmark suite includes controlled 2D temporal point clouds and biological single-cell time-course datasets.
The controlled datasets are designed to isolate geometric behavior in low-dimensional temporal data, including Ocean~\citep{shen2025multimarginal,petrovic2026curly} and a branching toy dataset, without relying on cell identities or velocity supervision.
The biological datasets test destructive snapshot settings in which individual cell identities are not shared across time, covering embryoid body differentiation~\citep{moon2019visualizing}, mouse reprogramming~\citep{schiebinger2019optimal}, induced trophoblast stem-cell reprogramming~\citep{liu2020reprogramming}, and multimodal CITE-seq/Multiome time courses~\citep{lance2022multimodal}.
Figure~\ref{fig:datasets-2d-overview} gives a visual overview of the 2D datasets used in the main low-dimensional experiments.
Detailed dataset sources and preprocessing choices are provided in Appendix~\ref{sec:app:benchmarks}.

\textbf{Baselines. }
We compare against representative transport- and flow-based trajectory models, covering action-based continuous dynamics~\citep{neklyudov2023action}, minibatch optimal-transport flow matching~\citep{tong2023conditional}, metric/geodesic flow matching~\citep{kapusniak2024metric}, pseudo-velocity-corrected flows~\citep{petrovic2026curly}, adversarial multi-marginal interpolant learning~\citep{kviman2026multimarginal}, and stochastic Schr\"odinger bridge dynamics~\citep{chen2023deep}.
Implementation details for baselines and PACE are provided in Appendices~\ref{sec:app:baselines} and~\ref{sec:app:pace-implementation}.
Qualitative prediction-versus-held-out overlays for the same 2D runs are shown in Appendix Figures~\ref{fig:app:stage2-ocean}--\ref{fig:app:stage2-ipsc}.

\subsection{Controlled 2D trajectories}

\begin{figure}[htp]
    \centering
    \begin{minipage}[t]{0.49\linewidth}
        \centering
        \includegraphics[width=\linewidth]{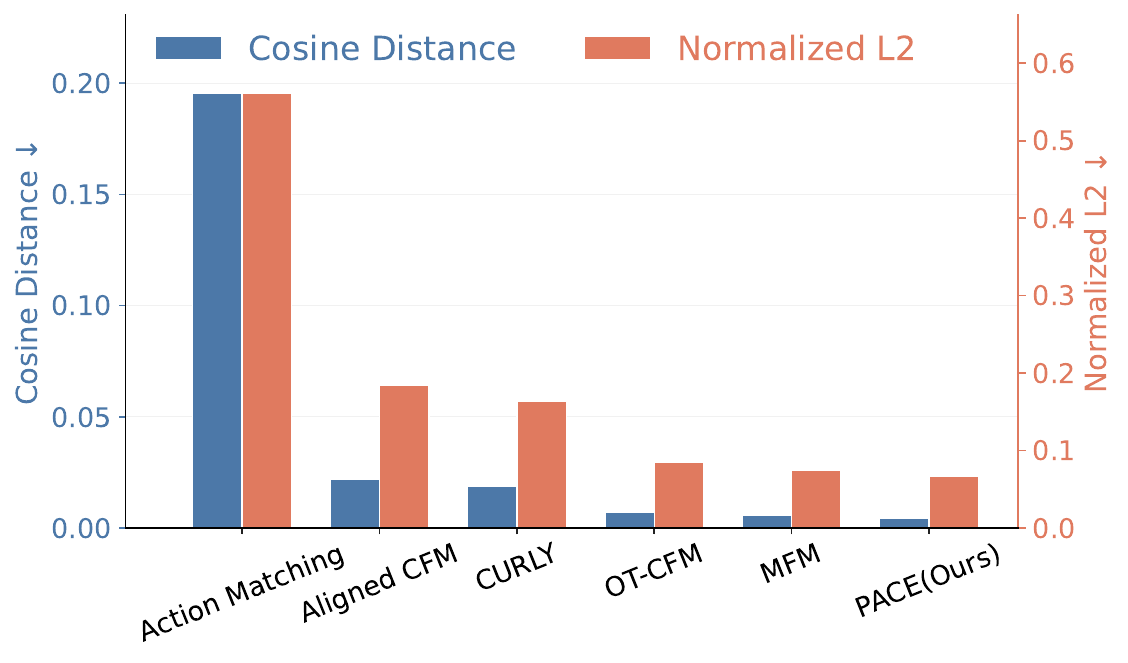}
        \centerline{\small (a) Ocean}
    \end{minipage}\hfill
    \begin{minipage}[t]{0.49\linewidth}
        \centering
        \includegraphics[width=\linewidth]{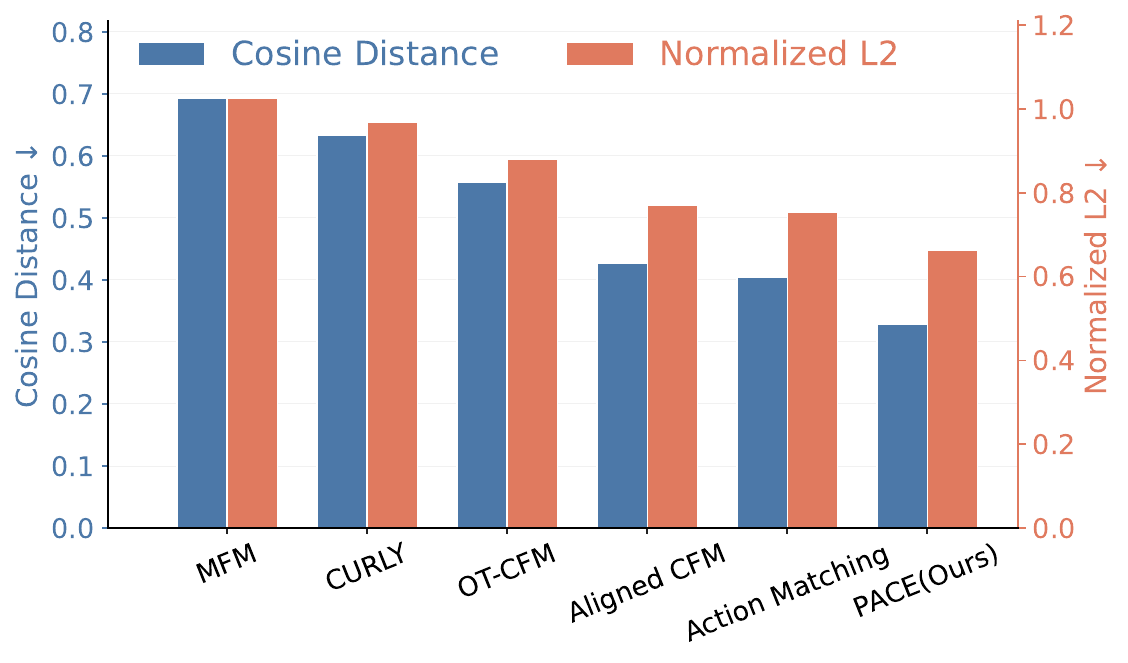}
        \centerline{\small (b) EB PHATE}
    \end{minipage}
    \caption{Velocity-alignment diagnostics on Ocean~\citep{petrovic2026curly} and EB PHATE~\citep{moon2019visualizing}. Ocean compares learned velocities with held-out simulator velocities, while EB PHATE compares learned 2D velocities with the RNA-velocity reference. Bars report cosine distance and normalized $L_2$ error; lower is better. }
    \label{fig:velocity-diagnostics} 
\end{figure}

We evaluate PACE on two controlled 2D settings, including Ocean, a rotational non-gradient benchmark from Curly Flow Matching~\citep{petrovic2026curly}, and a Two-Branch toy benchmark. Methods observe only unpaired point positions; simulator velocities and latent identities are not used for training.
For Ocean, we also hold out evaluation snapshots and compare the learned velocity field at those times with the simulator velocity field, testing whether the inferred dynamics recover the correct directionality rather than only matching held-out marginals.

Tables~\ref{tab:ocean_2d_1_3_5_7} and~\ref{tab:branch_toy_dim2} show that PACE gives the strongest overall held-out reconstruction. It is best on all Ocean metrics and remains strongest on most Two-Branch entries, with CURLY slightly better only for late-branch $\mathcal{W}_2$. The Ocean panel in Figure~\ref{fig:velocity-diagnostics} further checks the learned velocity direction against held-out simulator velocities; the per-timepoint diagnostic values are reported in Appendix Table~\ref{tab:ocean_holdout_velocity_metrics}. Together, these results indicate that the geometry-aware bias helps recover both smooth rotational motion and simple branching structure from snapshots alone.

\subsection{Low-dimensional single-cell trajectories}

We next ask whether the same behavior carries over to biological time courses represented in low-dimensional embeddings, where the true cell identities are unobserved and geometric structure must be inferred from destructive snapshots. 
Table~\ref{tab:single_cell_2d_time_average} summarizes the time-averaged low-dimensional biological results across EB PHATE, iPSC-Liu, and Schiebinger2019; the corresponding per-timepoint results are reported in Appendix Tables~\ref{tab:eb_phate_2d_3}, \ref{tab:ipsc_liu_2d_4_16}, and~\ref{tab:schiebinger2019_2d_6_11_16}. On EB PHATE, PACE obtains the best MMD, $\mathcal{W}_1$, and $\mathcal{W}_2$, showing that the method improves held-out marginal reconstruction in a manifold-aware PHATE embedding. The EB PHATE panel in Figure~\ref{fig:velocity-diagnostics} adds an independent directionality check by comparing the learned 2D velocity field at the held-out snapshot with RNA velocity, which is used only for evaluation. Appendix Table~\ref{tab:ebphate_holdout_velocity_metrics} 
shows that PACE also best aligns with this external velocity proxy, 
suggesting that the gain reflects both endpoint reconstruction and plausible local flow direction. 
On iPSC-Liu 2D, PACE is best on every reported metric. 
On Schiebinger2019, PACE is best on seven of nine detailed metric-timepoint entries and remains close on the remaining final-time scores. 
Overall, the low-dimensional biological results show that PACE transfers 
the controlled-geometry gains to both differentiation and reprogramming snapshots.
Compared with the strongest baseline for each dataset--metric pair, PACE reduces MMD, $\mathcal{W}_1$, and $\mathcal{W}_2$ by 36.2\%, 21.5\%, and 13.6\% on EB PHATE; by 16.1\%, 26.4\%, and 25.0\% on iPSC-Liu; and by 31.2\%, 25.5\%, and 22.4\% on Schiebinger2019.

\begin{table}[!htbp]
\centering
\scriptsize
\setlength{\tabcolsep}{3pt}
\renewcommand{\arraystretch}{0.88}
\caption{Time-averaged low-dimensional single-cell results. EB~\citep{moon2019visualizing} uses $t=3$, iPSC-Liu~\citep{liu2020reprogramming} uses $t\in\{4,16\}$, and Schiebinger2019~\citep{schiebinger2019optimal} uses $t\in\{6,11,16\}$.}
\label{tab:single_cell_2d_time_average}
\resizebox{\linewidth}{!}{
\begin{tabular}{lccccccccc}
\toprule
Method
& \multicolumn{3}{c}{EB PHATE}
& \multicolumn{3}{c}{iPSC-Liu}
& \multicolumn{3}{c}{Schiebinger2019} \\
\cmidrule(lr){2-4} \cmidrule(lr){5-7} \cmidrule(lr){8-10}
& MMD $\downarrow$ & $\mathcal{W}_1\downarrow$ & $\mathcal{W}_2\downarrow$
& MMD $\downarrow$ & $\mathcal{W}_1\downarrow$ & $\mathcal{W}_2\downarrow$
& MMD $\downarrow$ & $\mathcal{W}_1\downarrow$ & $\mathcal{W}_2\downarrow$ \\
\midrule
Action Matching
& 0.2436 & 0.4553 & 0.5273
& 0.5230 & 0.8005 & 0.9730
& 0.4057 & 0.4161 & 0.4754 \\
Aligned CFM
& \underline{0.1056} & \underline{0.2657} & \underline{0.3243}
& 0.5428 & \underline{0.6594} & \underline{0.8320}
& 0.3239 & 0.3175 & 0.3931 \\
CURLY
& 0.1384 & 0.3223 & 0.4422
& 0.5613 & 0.7186 & 0.9216
& 0.7687 & 1.3684 & 1.4660 \\
DMSB
& 0.1325 & 0.3053 & 0.4341
& 0.6456 & 1.0244 & 1.1023
& \underline{0.3030} & \underline{0.2774} & \underline{0.3662} \\
MFM
& 0.1518 & 0.3704 & 0.4451
& 0.5533 & 0.6657 & 0.8838
& 0.6990 & 1.5897 & 1.7942 \\
OT-CFM
& 0.1069 & 0.2784 & 0.3622
& \underline{0.5228} & 0.7166 & 0.9261
& 0.8072 & 1.9095 & 2.0235 \\
\midrule
\textbf{\model (ours)}
& \textbf{0.0674} & \textbf{0.2087} & \textbf{0.2803}
& \textbf{0.4385} & \textbf{0.4853} & \textbf{0.6238}
& \textbf{0.2084} & \textbf{0.2067} & \textbf{0.2840} \\
\bottomrule
\end{tabular}
}
\end{table}

\newpage
\subsection{Higher-dimensional single-cell benchmarks}

\noindent
\begin{minipage}[t]{0.53\linewidth}
\vspace{0pt}
High-dimensional PCA tests whether PACE still helps when Euclidean distances lose temporal contrast. On iPSC-Liu, PACE is best on all time-averaged 10D/50D metrics in Table~\ref{tab:ipsc_liu_dim10_50_mean}. On OP-Cite/OP-Multi 100D, it remains on the empirical Pareto front, leading OP-Cite MMD and OP-Multi $\mathcal{W}_1/\mathcal{W}_2$ in Appendix Table~\ref{tab:op_100d}. Figure~\ref{fig:norm-time-concentration-vs-dim} shows the diagnostic behind this regime, with norm CV approaching $0.3$ and inter-time separation approaching the within-time radius. The useful signal is therefore not global separation between time means, but local anisotropy within each snapshot; PACE turns that local geometry into coupling costs. This helps distinguish directionally consistent moves from distance-similar but geometrically implausible pairings. Appendix~\ref{sec:app:concentration} and Figures~\ref{fig:app:ipsc_dim}--\ref{fig:app:ipsc} provide the detailed concentration analysis.
\end{minipage}\hfill
\begin{minipage}[t]{0.44\linewidth}
\vspace{0pt}
\centering
\scriptsize
\setlength{\tabcolsep}{3pt}
\captionof{table}{iPSC-Liu~\citep{liu2020reprogramming} time-averaged results over holdouts $t\in\{4,8,12,16\}$.}
\label{tab:ipsc_liu_dim10_50_mean}
\begin{tabular}{llccc}
\toprule
Dim & Method & MMD $\downarrow$ & $\mathcal{W}_1\downarrow$ & $\mathcal{W}_2\downarrow$ \\
\midrule
10D & Action Matching & 0.5555 & 3.7801 & 4.3743 \\
 & Aligned CFM & \underline{0.4644} & \underline{2.9405} & \underline{3.3825} \\
 & CURLY & 0.5369 & 3.6128 & 4.2102 \\
 & DMSB & 0.7713 & 6.3088 & 6.4593 \\
 & MFM & 0.5273 & 3.4066 & 3.8382 \\
 & OT-CFM & 0.5145 & 3.3850 & 3.9447 \\
\cmidrule(lr){2-5}
 & \textbf{PACE (ours)} & \textbf{0.4095} & \textbf{2.9226} & \textbf{3.1879} \\
\midrule
50D & Action Matching & 0.3125 & 7.6126 & 8.4868 \\
 & Aligned CFM & \underline{0.2624} & \underline{7.3133} & \underline{8.0552} \\
 & CURLY & 0.3144 & 8.1144 & 8.9096 \\
 & DMSB & 0.5638 & 10.9368 & 11.2423 \\
 & MFM & 0.2853 & 7.7924 & 8.5454 \\
 & OT-CFM & 0.2898 & 7.8392 & 8.5954 \\
\cmidrule(lr){2-5}
 & \textbf{PACE (ours)} & \textbf{0.2550} & \textbf{6.7537} & \textbf{7.2836} \\
\bottomrule
\end{tabular}
\end{minipage}

\subsection{Ablation study}

We finally isolate the contribution of the main PACE components on the Schiebinger2019 holdouts used in Table~\ref{tab:schiebinger2019_2d_6_11_16}. Figure~\ref{fig:pace-ablation-per-tp} compares the full model with variants that remove coupling rematching, set the metric penalty to the Euclidean case, use all neighbors rather than local neighborhoods, or ablate the spatial and temporal kernel structure used to smooth local geometry. Each simplification degrades at least one metric or time point, and the full PACE variant gives the most stable low errors across MMD, $\mathcal{W}_1$, and $\mathcal{W}_2$. This indicates that the gains do not come from a single implementation detail: the anisotropic metric, local-neighborhood construction, spatial-temporal geometry smoothing, and rematching step all contribute to the final trajectory reconstruction.

\begin{center}
    \centering
    \begin{minipage}[t]{0.40\linewidth}
        \centering
        \vspace{0pt}
        \includegraphics[width=\linewidth]{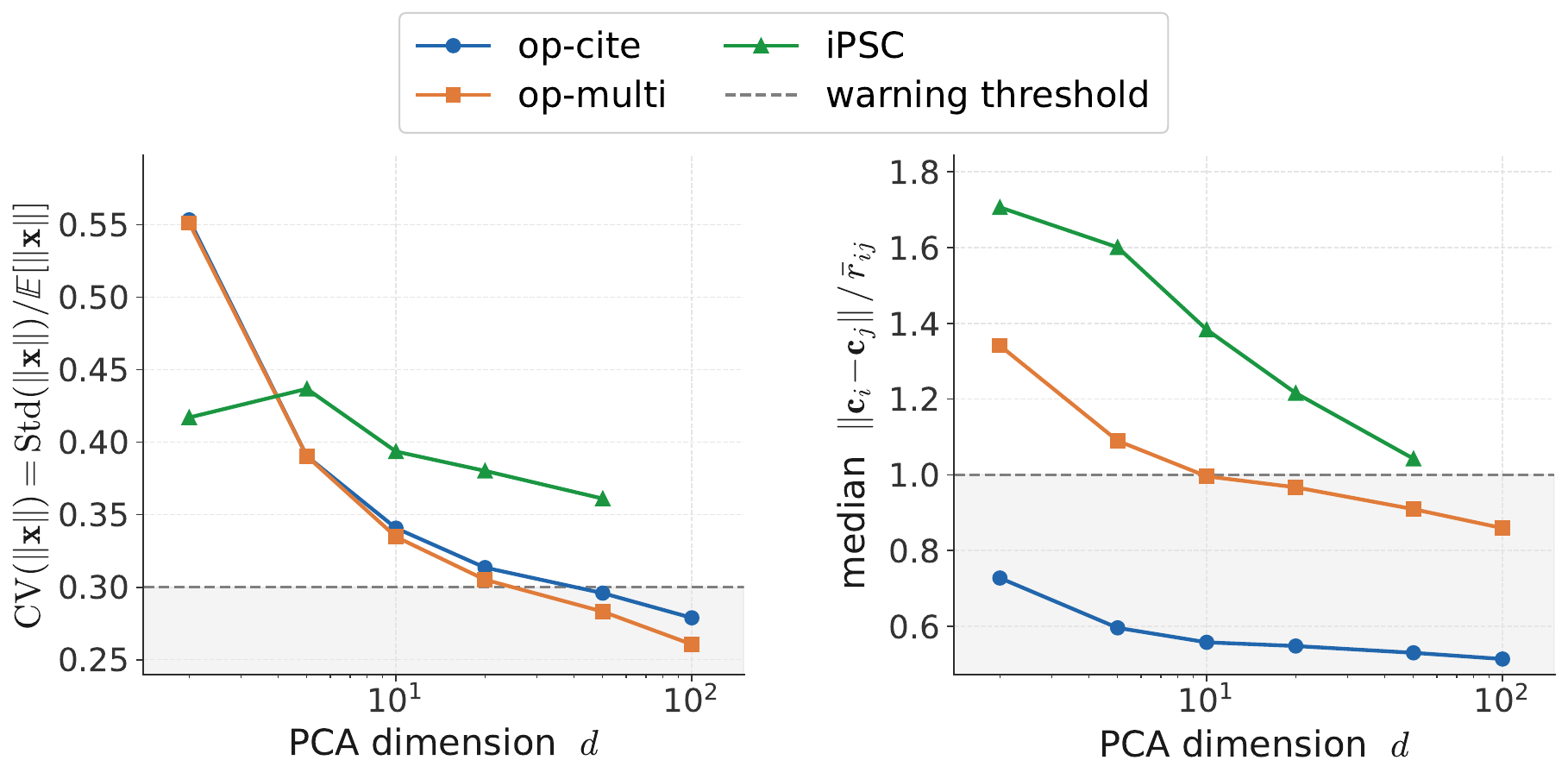}
        \vspace{-0.3em}
        \captionof{figure}{High-dimensional concentration diagnostics. Dashed lines mark norm CV $=0.3$ and inter-time/within-time ratio $=1.0$.}
        \label{fig:norm-time-concentration-vs-dim}
    \end{minipage}\hfill
    \begin{minipage}[t]{0.58\linewidth}
        \centering
        \vspace{0pt}
        \includegraphics[width=\linewidth]{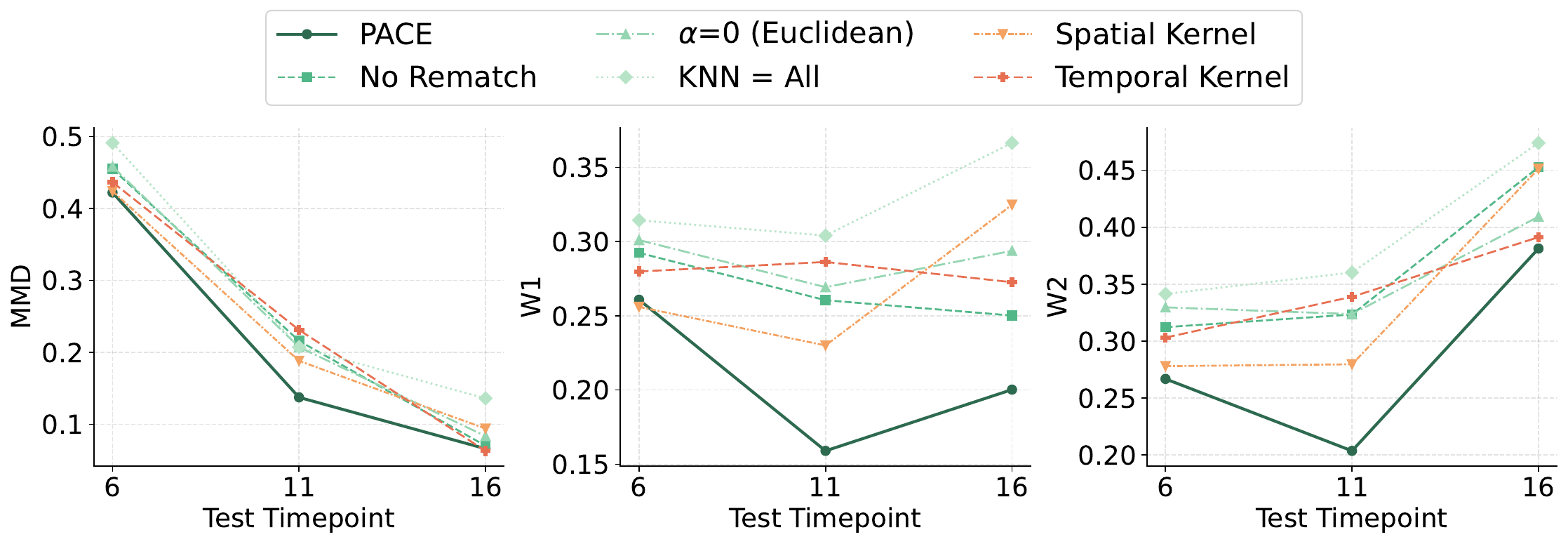}
        \vspace{-0.3em}
        \captionof{figure}{PACE ablations on Schiebinger2019~\citep{schiebinger2019optimal} holdouts. Lower is better; variants remove rematching, use Euclidean action ($\alpha=0$), use all neighbors, or drop spatial/temporal kernels.}
        \label{fig:pace-ablation-per-tp}
    \end{minipage}
\end{center}

\section{Conclusion}
In this paper, we have introduced PACE, a geometry-aware framework for trajectory inference from destructive single-cell time-course snapshots.
PACE replaces fixed Euclidean endpoint costs with an anisotropic path cost induced by a local spatiotemporal metric, then alternates bridge learning with coupling refinement and distills the resulting dynamics into a velocity field.
Across controlled and biological benchmarks, the results support local geometry as a useful inductive bias for held-out snapshot reconstruction without paired cells, lineage tracing, or velocity supervision; the velocity diagnostics further show that the learned flows recover plausible directionality when external velocity references are available only for evaluation.
Ablations show that this behavior depends on the combination of anisotropic local costs, neighborhood-restricted geometry, smoothing, and coupling rematching rather than a single implementation detail.
Future work will focus on more robust metric estimation, uncertainty quantification, and partially supervised or multimodal time-course settings.

\begin{ack}
We thank Yuchen Yang, Minsi Ren, Zhuo Xu, Yucheng Luo, and Mingzheng Fang for discussions and for providing feedback on our manuscript.

We also gratefully acknowledge the support of Hangzhou West Lake Pearl Program, Westlake University Research Center for Industries of the Future, and Westlake University Center for High-performance Computing. The content is solely the responsibility of the authors and does not necessarily represent the official views of the funding entities.

\end{ack}

\bibliographystyle{unsrt}
\bibliography{references}

\appendix
\clearpage
\section*{Technical appendices and supplementary material}
\setcounter{figure}{0}
\setcounter{table}{0}
\renewcommand{\thefigure}{A.\arabic{figure}}
\renewcommand{\thetable}{A.\arabic{table}}
\renewcommand{\theHfigure}{appendix.\arabic{figure}}
\renewcommand{\theHtable}{appendix.\arabic{table}}
\begin{enumerate}[label=\Alph*.]
  \item \nameref{sec:app:benchmarks} \dotfill \pageref{sec:app:benchmarks}
  \item \nameref{sec:app:baselines} \dotfill \pageref{sec:app:baselines}
  \item \nameref{sec:app:pace-implementation} \dotfill \pageref{sec:app:pace-implementation}
  \item \nameref{app:illposed} \dotfill \pageref{app:illposed}
  \item \nameref{app:metric-correspondence} \dotfill \pageref{app:metric-correspondence}
  \item \nameref{app:optimization-monotonicity} \dotfill \pageref{app:optimization-monotonicity}
  \item \nameref{app:regularizers} \dotfill \pageref{app:regularizers}
  \item \nameref{sec:app:concentration} \dotfill \pageref{sec:app:concentration}
  \item \nameref{sec:app:ipsc} \dotfill \pageref{sec:app:ipsc}
  \item \nameref{app:alternating-intuition} \dotfill \pageref{app:alternating-intuition}
  \item \nameref{app:adaptive-projectors} \dotfill \pageref{app:adaptive-projectors}
  \item \nameref{app:bandwidth} \dotfill \pageref{app:bandwidth}
  \item \nameref{app:limitations} \dotfill \pageref{app:limitations}
  \item \nameref{app:efficiency} \dotfill \pageref{app:efficiency}
\end{enumerate}

\vspace{1em}

\section{Benchmark and Metric Details}
\label{sec:app:benchmarks}

All datasets are cast as unpaired time-course point clouds. A point is a cell or synthetic particle state, and a time point is an empirical marginal distribution. Methods receive only the training snapshots; held-out snapshots are used only for evaluation. No method receives cell identities, cross-time correspondences, RNA velocity, or any ground-truth trajectory pairing.

\paragraph{Shared split and rollout protocol.}
For a held-out label $t_h$, the evaluator finds the nearest observed training labels $t_a<t_h<t_b$. The learned velocity field is rolled out from the empirical source cloud at $t_a$ toward $t_b$ on the same normalized training-time scale used during training. We use 101 equally spaced integration grid points (100 Euler updates) and extract the intermediate frame at ratio $(t_h-t_a)/(t_b-t_a)$. The predicted cloud is then compared with the empirical held-out cloud at $t_h$. Thus, the reported numbers measure recovery of the missing marginal distribution, not recovery of unobserved cell identities.

\clearpage
\begingroup
\centering
\captionof{table}{Dataset sources, representations, and held-out splits used in Section~\ref{sec:experiment}. Controlled 2D coordinates are used directly. Biological embeddings are standardized by fitting a \texttt{StandardScaler} on training snapshots only, then applying it to train and held-out snapshots.}
\label{tab:app:datasets}
\small
\setlength{\tabcolsep}{3pt}
\begin{tabular}{p{0.16\linewidth}p{0.31\linewidth}p{0.25\linewidth}p{0.20\linewidth}}
\toprule
Benchmark & Source and role & Representation and preprocessing & Train/test labels \\
\midrule
Ocean (Gulf of Mexico vortex) & HYCOM-derived particle benchmark from Shen et al.~\citep{shen2025multimarginal}, also used in CURLY~\citep{petrovic2026curly}. Training uses only particle positions; velocities are not used for supervision. & \path{data/oceans/oceans.npz}; \texttt{positions} has shape $(9,111,2)$. No whitening; full frames are used. & Train $\{0,2,4,6,8\}$; test $\{1,3,5,7\}$. \\
\addlinespace
Two-Branch & In-house synthetic bifurcation generated by \path{src/data_preprocess/generate_two_branch_data.py}. It tests recovery of a trunk that splits into two branches. & Six 2D frames, 64 points per frame, with latent pseudotime and branch labels stored for diagnostics but not given to methods. No whitening; full frames are used. & Train $\{0,2,3,5\}$; test $\{1,4\}$. \\
\addlinespace
EB PHATE & Embryoid body differentiation data from PHATE~\citep{moon2019visualizing}. It represents a low-dimensional differentiation trajectory. & \path{data/eb_velocity_v5.npz}; 16,819 cells with \texttt{phate} $(N,2)$ and \texttt{pcs} $(N,100)$. The 2D experiment uses PHATE and fits whitening on training cells only. & Train $\{0,1,2,4\}$; test $\{3\}$. \\
\addlinespace
Schiebinger2019 & Mouse reprogramming time course from Schiebinger et al.~\citep{schiebinger2019optimal}, converted from the scNODE data package using serum-filtered FLE coordinates. & Local H5AD file; \texttt{X} stores 2D FLE coordinates and \texttt{obs["day"]} stores integer day labels. Whitening is fit on training cells only. & Train all integer days $0,\ldots,18$ except $\{6,11,16\}$; test $\{6,11,16\}$. \\
\addlinespace
iPSC-Liu & Human reprogramming time course from Liu et al.~\citep{liu2020reprogramming}. It is used to evaluate missing-timepoint reconstruction in PCA representations. & \path{data/iPSC_liu_pca.h5ad}; \texttt{obsm["X\_pca"]} provides PCA scores and \texttt{obs["timepoint"]} stores labels such as \texttt{D0}, parsed as integers. Whitening is fit on training cells only. & 2D: train $\{0,8,12,20,24\}$, test $\{4,16\}$. 10D/50D: separate single-holdout runs for $t\in\{4,8,12,16\}$. \\
\addlinespace
OP-Cite/OP-Multi & NeurIPS Open Problems multimodal single-cell integration challenge data~\citep{lance2022multimodal}. OP-Cite uses CITE-seq features; OP-Multi uses Multiome features. & \path{op_cite_inputs_0.h5ad} and \path{op_train_multi_targets_0.h5ad}; both use \texttt{obsm["X\_pca"]}. Table~\ref{tab:op_100d} uses all 100 PCs, with whitening fit on training cells only. & Train $\{2,4,7\}$; test $\{3\}$. \\
\bottomrule
\end{tabular}
\par
\endgroup

\paragraph{Ocean benchmark.}
The Ocean benchmark follows the Gulf of Mexico vortex experiment of Shen et al.~\citep{shen2025multimarginal}. The data are generated by first extracting a velocity field around a vortex feature from high-resolution HYbrid Coordinate Ocean Model (HYCOM) reanalysis data~\citep{chassignet2007hycom}, then simulating particles representing buoys or ocean debris through that field. The original benchmark contains approximately 1000 observations across five training times and four validation times; the released array used here has nine snapshots with 111 particles per snapshot. Unlike the reference-family model in Shen et al., PACE and all baselines in our comparison receive only the particle positions at training snapshots, with no velocity supervision. The HYCOM-derived velocity field is held out from training and used only for the velocity-alignment diagnostic in Table~\ref{tab:ocean_holdout_velocity_metrics}.

\paragraph{Velocity-alignment diagnostics.}
Figure~\ref{fig:velocity-diagnostics} summarizes the velocity-alignment checks used in the main experiments. The tables below report the numerical values behind the two panels: Ocean compares learned velocities with held-out simulator velocities, while EB PHATE compares learned 2D velocities with the RNA-velocity reference. These reference velocities are used only for diagnostics, not for training.

\begin{table}[hbtp]
\centering
\begin{minipage}{0.75\linewidth}
\centering
\caption{Branching toy (2D) results at held-out time points $t\in\{1,4\}$.}
\label{tab:branch_toy_dim2}
\small
\setlength{\tabcolsep}{3pt}
\renewcommand{\arraystretch}{0.82}
\begin{tabular}{lcccccc}
\toprule
Method & \multicolumn{2}{c}{MMD $\downarrow$} & \multicolumn{2}{c}{$\mathcal{W}_1\downarrow$} & \multicolumn{2}{c}{$\mathcal{W}_2\downarrow$} \\
\cmidrule(lr){2-3} \cmidrule(lr){4-5} \cmidrule(lr){6-7}
 & $t=1$ & $t=4$ & $t=1$ & $t=4$ & $t=1$ & $t=4$ \\
\midrule
Action Matching & 0.6129 & 0.2905 & 0.3343 & 0.2732 & 0.3419 & 0.3048 \\
Aligned CFM & 0.6524 & 0.3439 & 0.3874 & 0.309 & 0.3944 & 0.3371 \\
CURLY & \underline{0.2664} & \underline{0.0767} & 0.1703 & \underline{0.1797} & 0.2123 & \textbf{0.2087} \\
DMSB & 1.0232 & 1.1585 & 4.5832 & 2.8358 & 4.5845 & 2.8406 \\
MFM & 0.3817 & 0.1145 & \underline{0.1686} & 0.1865 & \underline{0.1832} & \underline{0.2116} \\
OT-CFM & 0.5279 & 0.1244 & 0.2433 & 0.1841 & 0.2509 & 0.2119 \\
\midrule
\textbf{\model (ours)} & \textbf{0.1968} & \textbf{0.0146} & \textbf{0.1386} & \textbf{0.1752} & \textbf{0.1686} & 0.2236 \\
\bottomrule
\end{tabular}
\end{minipage}
\end{table}

\begingroup
\centering
\captionof{table}{Ocean~\citep{petrovic2026curly} velocity-alignment diagnostics on held-out time points. The simulator velocity is used only for evaluation. Cosine distance is $1-\cos(\hat v, v)$, and normalized $L_2$ compares unit-normalized velocity directions; lower is better.}
\label{tab:ocean_holdout_velocity_metrics}
\scriptsize
\setlength{\tabcolsep}{2.5pt}
\resizebox{\linewidth}{!}{
\begin{tabular}{lcccccccccc}
\toprule
Method & \multicolumn{5}{c}{Cosine distance $\downarrow$} & \multicolumn{5}{c}{Normalized $L_2\downarrow$} \\
\cmidrule(lr){2-6} \cmidrule(lr){7-11}
 & $t=1$ & $t=3$ & $t=5$ & $t=7$ & Avg. & $t=1$ & $t=3$ & $t=5$ & $t=7$ & Avg. \\
\midrule
Action Matching & 0.3535 & 0.0408 & 0.1974 & 0.1900 & 0.1954 & 0.8174 & 0.2317 & 0.5965 & 0.5993 & 0.5612 \\
Aligned CFM & 0.0440 & 0.0269 & 0.0101 & 0.0073 & 0.0221 & 0.2808 & 0.2118 & 0.1311 & 0.1122 & 0.1840 \\
CURLY & 0.0227 & 0.0068 & 0.0220 & 0.0231 & 0.0186 & 0.1997 & \underline{0.0906} & 0.1709 & 0.1955 & 0.1642 \\
MFM & \textbf{0.0029} & \underline{0.0063} & 0.0089 & \underline{0.0043} & \underline{0.0056} & \textbf{0.0417} & 0.0913 & 0.0965 & \underline{0.0698} & \underline{0.0748} \\
OT-CFM & \underline{0.0040} & 0.0101 & \underline{0.0088} & 0.0058 & 0.0072 & \underline{0.0738} & 0.1147 & \underline{0.0731} & 0.0801 & 0.0854 \\
\midrule
\textbf{\model (ours)} & 0.0041 & \textbf{0.0030} & \textbf{0.0069} & \textbf{0.0031} & \textbf{0.0043} & 0.0756 & \textbf{0.0631} & \textbf{0.0671} & \textbf{0.0633} & \textbf{0.0673} \\
\bottomrule
\end{tabular}
}
\par
\endgroup

\vspace{0.75em}
\begingroup
\centering
\captionof{table}{EB PHATE~\citep{moon2019visualizing} velocity-alignment diagnostics at the held-out snapshot $t=3$. The RNA-velocity reference is used only for evaluation. Cosine distance is $1-\cos(\hat v, v)$, and normalized $L_2$ compares unit-normalized velocity directions; lower is better.}
\label{tab:ebphate_holdout_velocity_metrics}
\small
\setlength{\tabcolsep}{4pt}
\resizebox{0.65\linewidth}{!}{
\begin{tabular}{lcc}
\toprule
Method & Cosine distance $\downarrow$ & Normalized $L_2\downarrow$ \\
\midrule
Action Matching & \underline{0.4041} & \underline{0.7542} \\
Aligned CFM & 0.4273 & 0.7707 \\
CURLY & 0.6325 & 0.9693 \\
MFM & 0.6930 & 1.0261 \\
OT-CFM & 0.5571 & 0.8807 \\
\midrule
\textbf{\model (ours)} & \textbf{0.3287} & \textbf{0.6630} \\
\bottomrule
\end{tabular}
}
\par
\endgroup

\paragraph{Low-dimensional single-cell reconstruction results.}
Table~\ref{tab:single_cell_2d_time_average} reports the time-averaged summary used in the main text. The tables below give the corresponding per-timepoint results for each low-dimensional biological benchmark: EB PHATE has one held-out PHATE snapshot, iPSC-Liu has two held-out 2D PCA snapshots, and Schiebinger2019 has three held-out reprogramming days. These detailed tables show whether the averaged gains are consistent across individual held-out times.

\paragraph{Subsampling and dimensionality.}
For biological datasets, \texttt{samples\_per\_timepoint} is either an experiment-specific cap or \texttt{null}. When it is a cap, each selected time point is sampled without replacement using the experiment seed before whitening. When it is \texttt{null}, all available cells from the selected time points are used. For EB, the raw timepoint counts are $2381,4163,3278,3665,3332$ across labels $0,\ldots,4$. For iPSC-Liu and OP-Cite/OP-Multi, the dimensionality is chosen by taking the first $d$ PCA coordinates from \texttt{X\_pca}. The 10D and 50D iPSC-Liu tables aggregate separate single-holdout runs; in each run the held-out label is removed from the training label set.

\paragraph{Metrics.}
For every held-out time point, we compare the predicted point cloud with the observed point cloud using three distributional metrics. MMD is computed with a Gaussian kernel and a median-heuristic bandwidth, and the reported value is the square root of the nonnegative MMD estimate. $\mathcal{W}_1$ uses uniform empirical weights and the Euclidean ground cost between predicted and observed samples. $\mathcal{W}_2$ uses uniform empirical weights and squared Euclidean ground cost, then reports the square root of the optimal transport value. These metrics are distributional: they evaluate reconstruction of the held-out marginal distribution, not cell-wise recovery of unobserved identities.

\begin{table}[hbtp]
\centering
\begin{minipage}{0.6\linewidth}
\centering
\caption{Per-timepoint results on EB PHATE~\citep{moon2019visualizing} (2D) with held-out $t=3$.}
\label{tab:eb_phate_2d_3}

\small
\setlength{\tabcolsep}{3pt}

\begin{tabular}{lccc}
\toprule
Method & MMD $\downarrow$ & $\mathcal{W}_1\downarrow$ & $\mathcal{W}_2\downarrow$ \\
\cmidrule(lr){2-2} \cmidrule(lr){3-3} \cmidrule(lr){4-4}
 & $t=3$ & $t=3$ & $t=3$ \\
\midrule
Action Matching & 0.2436 & 0.4553 & 0.5273 \\
Aligned CFM & \underline{0.1056} & \underline{0.2657} & \underline{0.3243} \\
CURLY & 0.1384 & 0.3223 & 0.4422 \\
DMSB & 0.1325 & 0.3053 & 0.4341 \\
MFM & 0.1518 & 0.3704 & 0.4451 \\
OT-CFM & 0.1069 & 0.2784 & 0.3622 \\
\midrule
\textbf{\model (ours)} & \textbf{0.0674} & \textbf{0.2087} & \textbf{0.2803} \\
\bottomrule
\end{tabular}

\end{minipage}
\end{table}

\vspace{0.75em}
\begin{table}[hbtp]
\centering
\begin{minipage}{0.85\linewidth}
\centering

\caption{Per-timepoint results on iPSC-Liu~\citep{liu2020reprogramming} (2D) with held-out time points $t\in\{4,16\}$.}
\label{tab:ipsc_liu_2d_4_16}

\small
\setlength{\tabcolsep}{4pt}

\begin{tabular}{lcccccc}
\toprule
Method & \multicolumn{2}{c}{MMD $\downarrow$} & \multicolumn{2}{c}{$\mathcal{W}_1\downarrow$} & \multicolumn{2}{c}{$\mathcal{W}_2\downarrow$} \\
\cmidrule(lr){2-3} \cmidrule(lr){4-5} \cmidrule(lr){6-7}
 & $t=4$ & $t=16$ & $t=4$ & $t=16$ & $t=4$ & $t=16$ \\
\midrule
Action Matching & \underline{0.7369} & 0.3090 & 0.9719 & 0.6290 & 1.0726 & \underline{0.8733} \\
Aligned CFM & 0.8265 & \underline{0.2590} & 0.7781 & 0.5407 & \underline{0.7858} & 0.8781 \\
CURLY & 0.8156 & 0.3069 & 0.9030 & \underline{0.5341} & 0.9292 & 0.9140 \\
DMSB & 0.8663 & 0.4248 & \underline{0.7176} & 1.3311 & 0.7986 & 1.4060 \\
MFM & 0.8334 & 0.2731 & 0.7646 & 0.5667 & 0.7860 & 0.9815 \\
OT-CFM & 0.7787 & 0.2668 & 0.8695 & 0.5636 & 0.8983 & 0.9539 \\
\midrule
\textbf{\model (ours)} & \textbf{0.6808} & \textbf{0.1962} & \textbf{0.5579} & \textbf{0.4126} & \textbf{0.5646} & \textbf{0.6830} \\
\bottomrule
\end{tabular}

\end{minipage}
\end{table}

\vspace{0.75em}
\begin{table}[hbtp]
\centering
\caption{Per-timepoint results on Schiebinger2019~\citep{schiebinger2019optimal}, holdout $t\in\{6,11,16\}$.}
\label{tab:schiebinger2019_2d_6_11_16}
\small
\setlength{\tabcolsep}{4pt}
\resizebox{1\linewidth}{!}{
\begin{tabular}{lccccccccc}
\toprule
Method & \multicolumn{3}{c}{MMD $\downarrow$} & \multicolumn{3}{c}{$\mathcal{W}_1\downarrow$} & \multicolumn{3}{c}{$\mathcal{W}_2\downarrow$} \\
\cmidrule(lr){2-4} \cmidrule(lr){5-7} \cmidrule(lr){8-10}
 & $t=6$ & $t=11$ & $t=16$ & $t=6$ & $t=11$ & $t=16$ & $t=6$ & $t=11$ & $t=16$ \\
\midrule
Action Matching & \underline{0.5056} & 0.4940 & 0.2174 & 0.3203 & 0.4982 & 0.4298 & 0.3283 & 0.5337 & 0.5642 \\
Aligned CFM & 0.5712 & 0.3405 & \textbf{0.0601} & 0.4126 & 0.3173 & 0.2228 & 0.4291 & \underline{0.3737} & \textbf{0.3764} \\
CURLY & 0.9308 & 1.0201 & 0.3553 & 1.4737 & 1.9135 & 0.7182 & 1.5113 & 1.9386 & 0.9481 \\
DMSB & 0.5081 & \underline{0.3196} & 0.0813 & \underline{0.3118} & \underline{0.3108} & \underline{0.2096} & \underline{0.3174} & 0.3760 & 0.4051 \\
MFM & 0.8676 & 0.8858 & 0.3437 & 2.6083 & 1.4602 & 0.7005 & 3.0265 & 1.4793 & 0.8768 \\
OT-CFM & 1.1550 & 1.0005 & 0.2661 & 3.1652 & 1.8933 & 0.6701 & 3.2608 & 1.9239 & 0.8856 \\
\midrule
\textbf{\model (ours)} & \textbf{0.4219} & \textbf{0.1373} & \underline{0.0660} & \textbf{0.2609} & \textbf{0.1590} & \textbf{0.2003} & \textbf{0.2669} & \textbf{0.2037} & \underline{0.3813} \\
\bottomrule
\end{tabular}
}
\end{table}

\clearpage
\paragraph{Qualitative Stage-2 trajectory visualizations.}
Figures~\ref{fig:app:stage2-ocean}--\ref{fig:app:stage2-ipsc} complement the quantitative tables with selected 2D Stage-2 rollouts. For each dataset, the trajectory panel shows the learned rollout geometry for each baseline, while the prediction panel overlays the predicted held-out marginal with the empirical held-out snapshot. These plots are qualitative diagnostics only; all models are trained from unpaired snapshots without cell identities or ground-truth trajectories.

\begin{figure}[p]
    \centering
    \includegraphics[width=\linewidth]{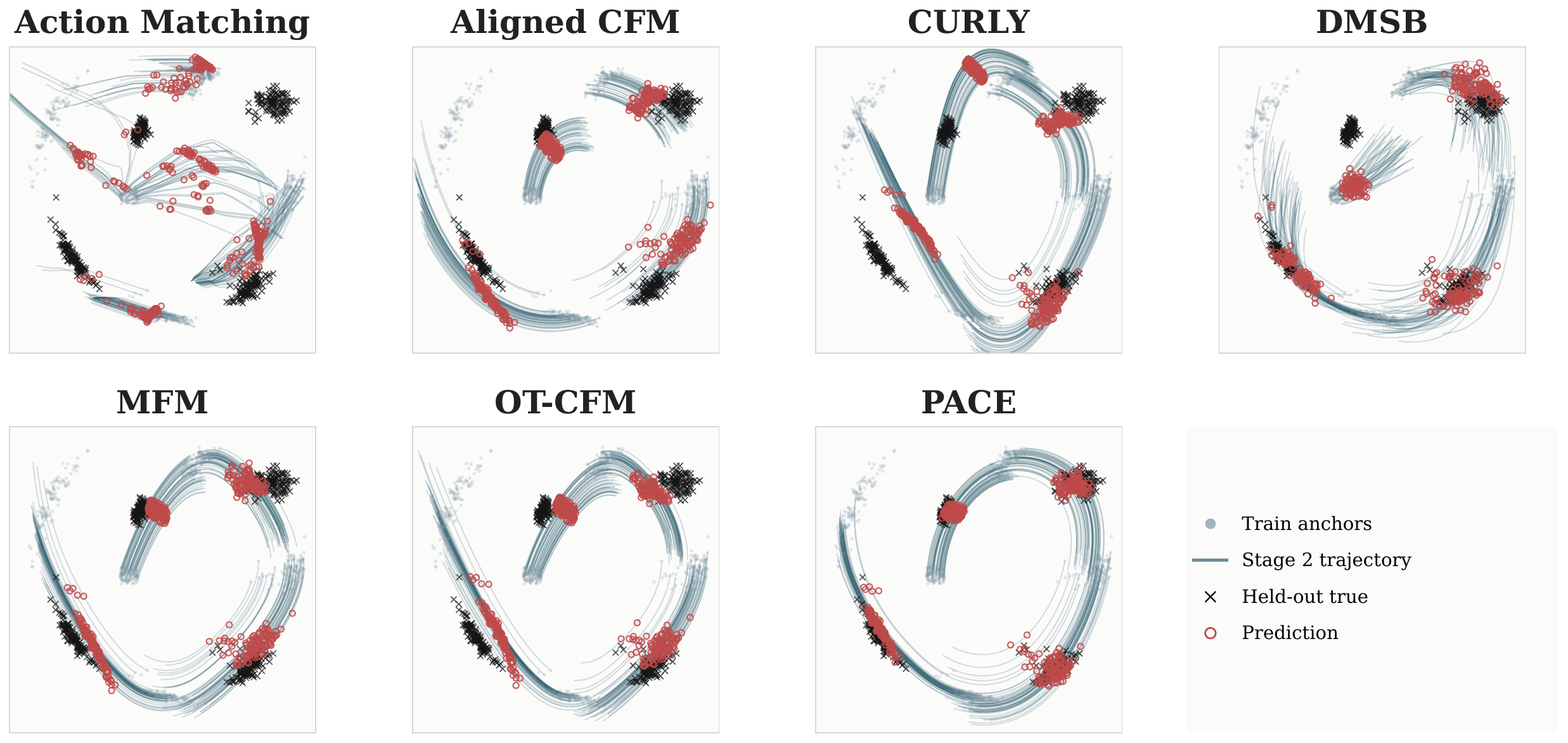}
    \vspace{0.4em}
    \includegraphics[width=\linewidth]{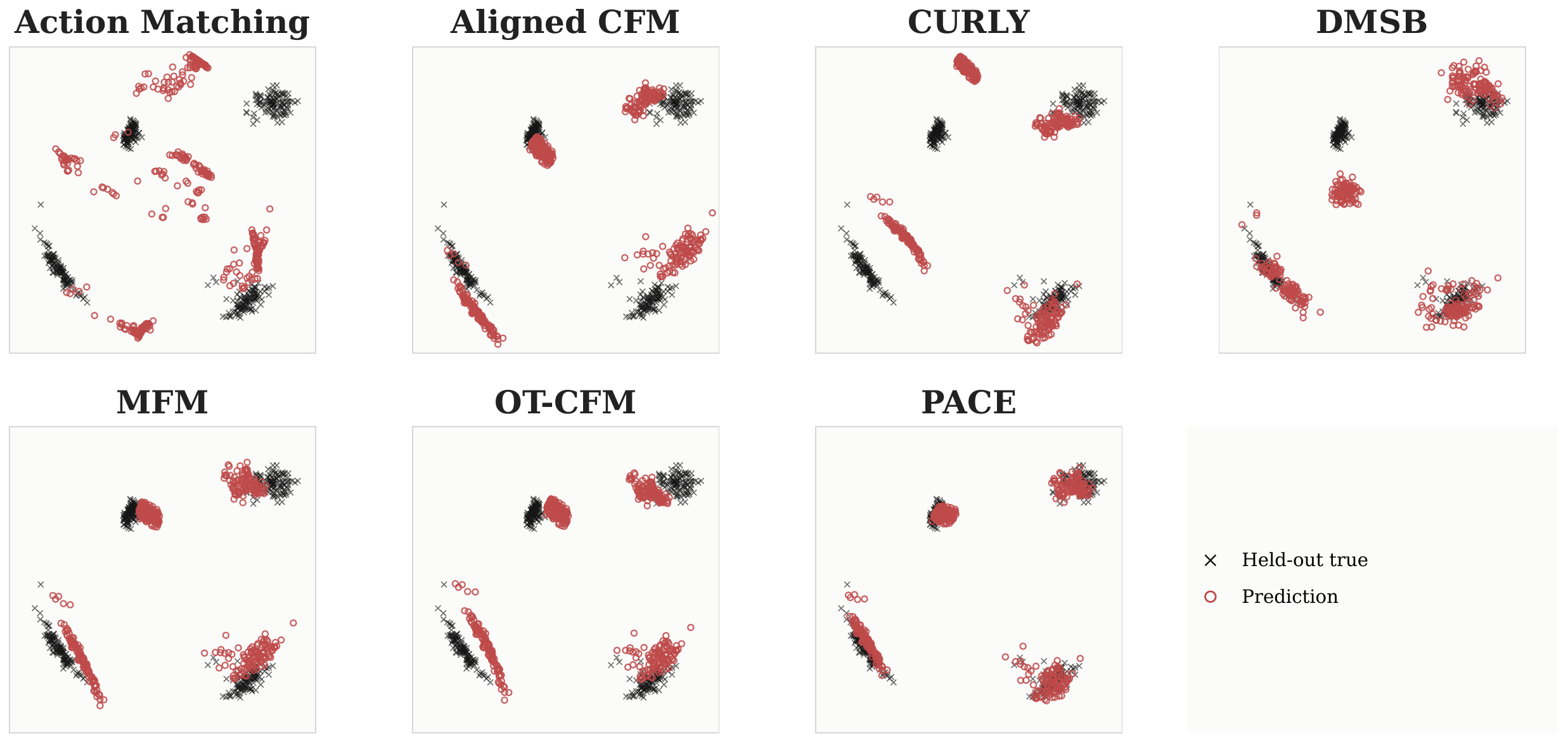}
    \caption{Qualitative Stage-2 results on Ocean~\citep{shen2025multimarginal,petrovic2026curly}. The top panel shows baseline rollouts; the bottom panel compares predicted and held-out point clouds.}
    \label{fig:app:stage2-ocean}
\end{figure}

\begin{figure}[p]
    \centering
    \includegraphics[width=\linewidth]{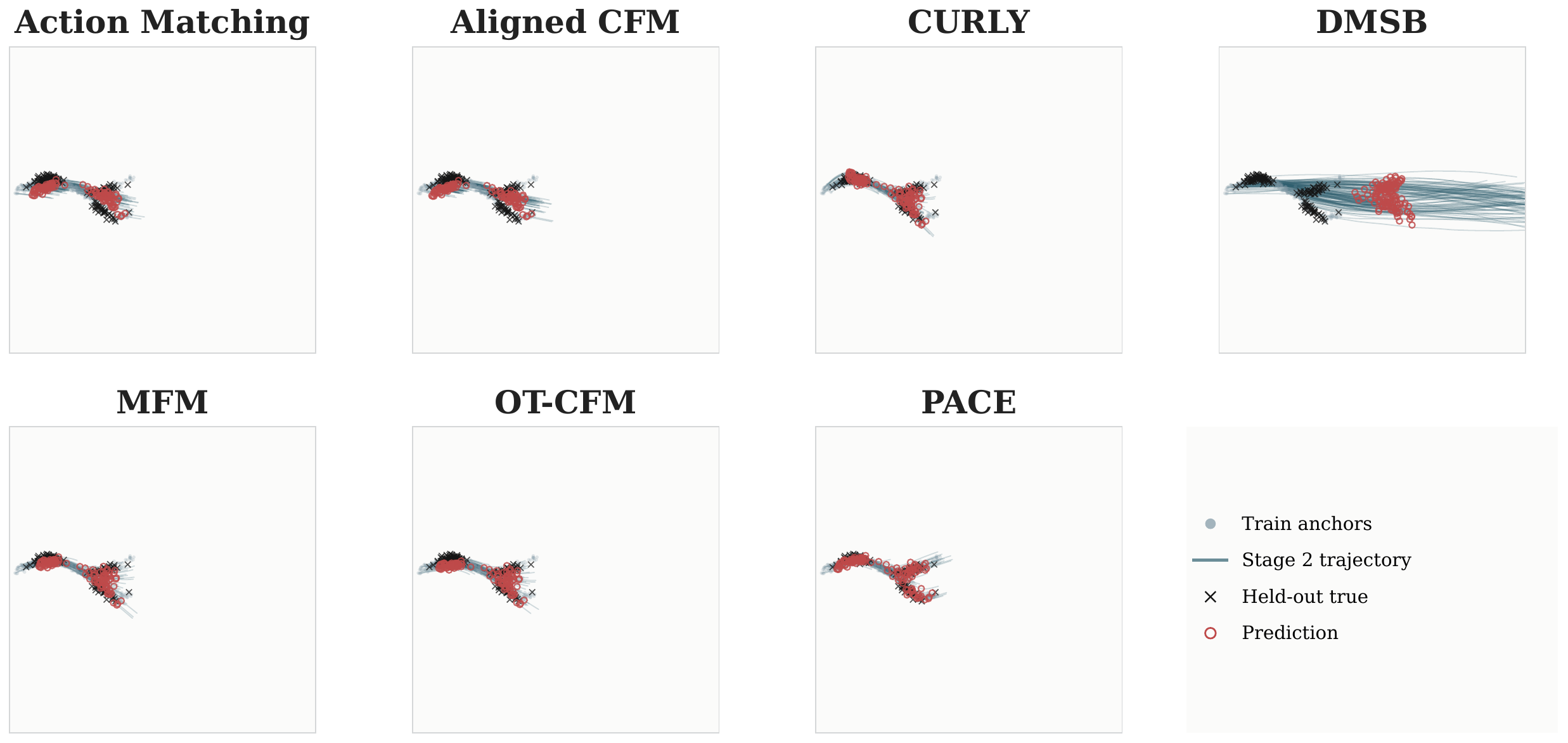}
    \vspace{0.4em}
    \includegraphics[width=\linewidth]{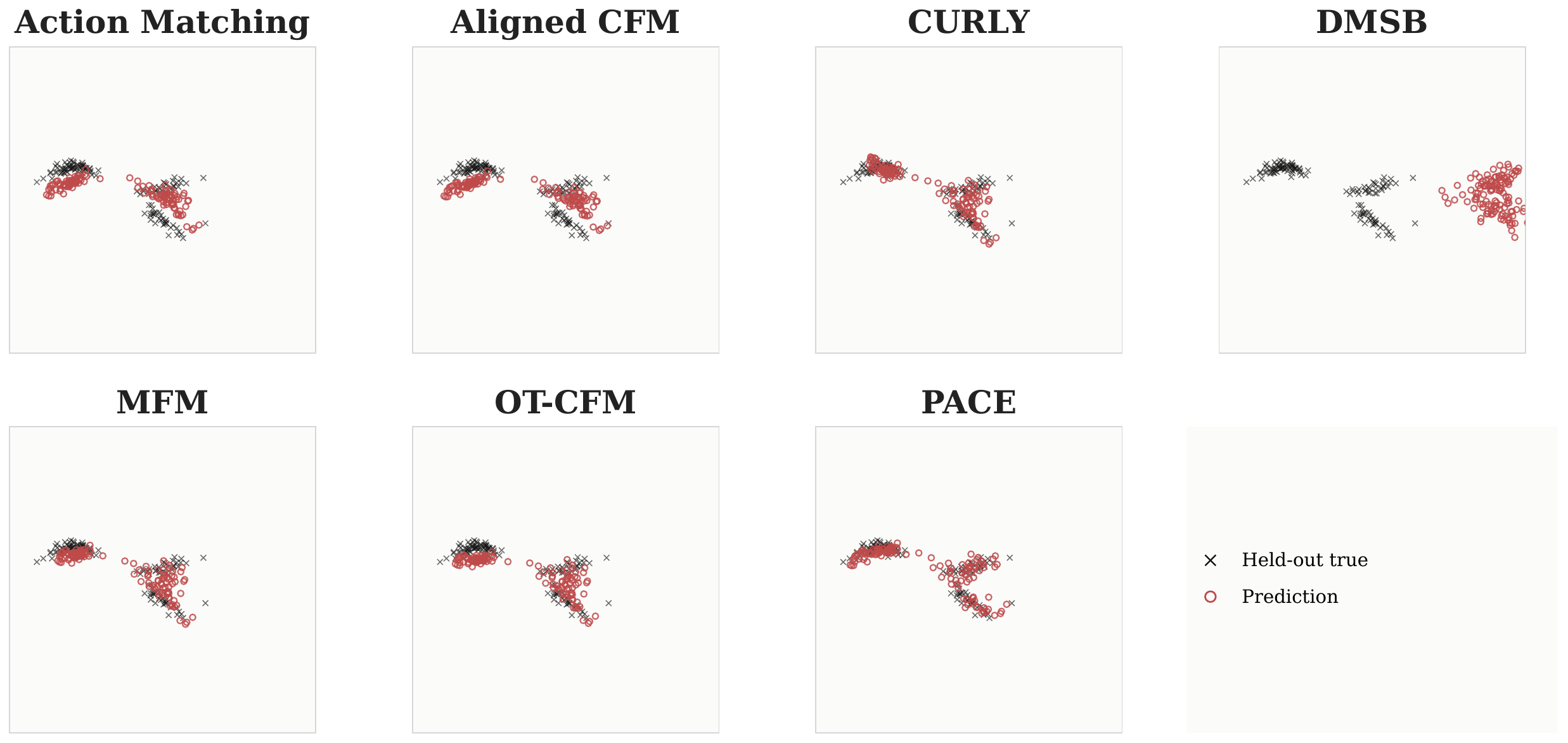}
    \caption{Qualitative Stage-2 results on the Two-Branch toy benchmark. The top panel shows baseline rollouts; the bottom panel compares predicted and held-out point clouds.}
    \label{fig:app:stage2-two-branch}
\end{figure}

\begin{figure}[p]
    \centering
    \includegraphics[width=\linewidth]{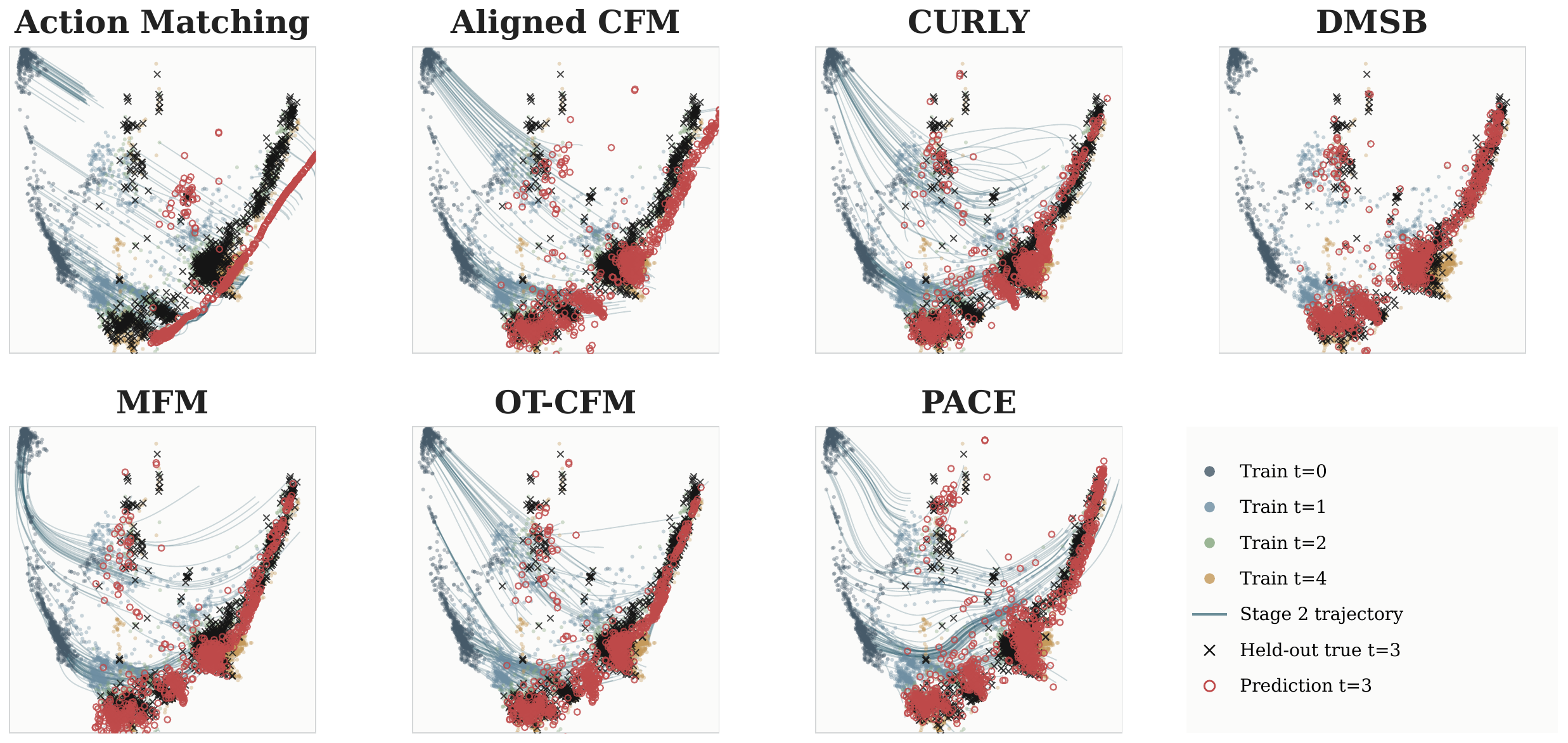}
    \vspace{0.4em}
    \includegraphics[width=\linewidth]{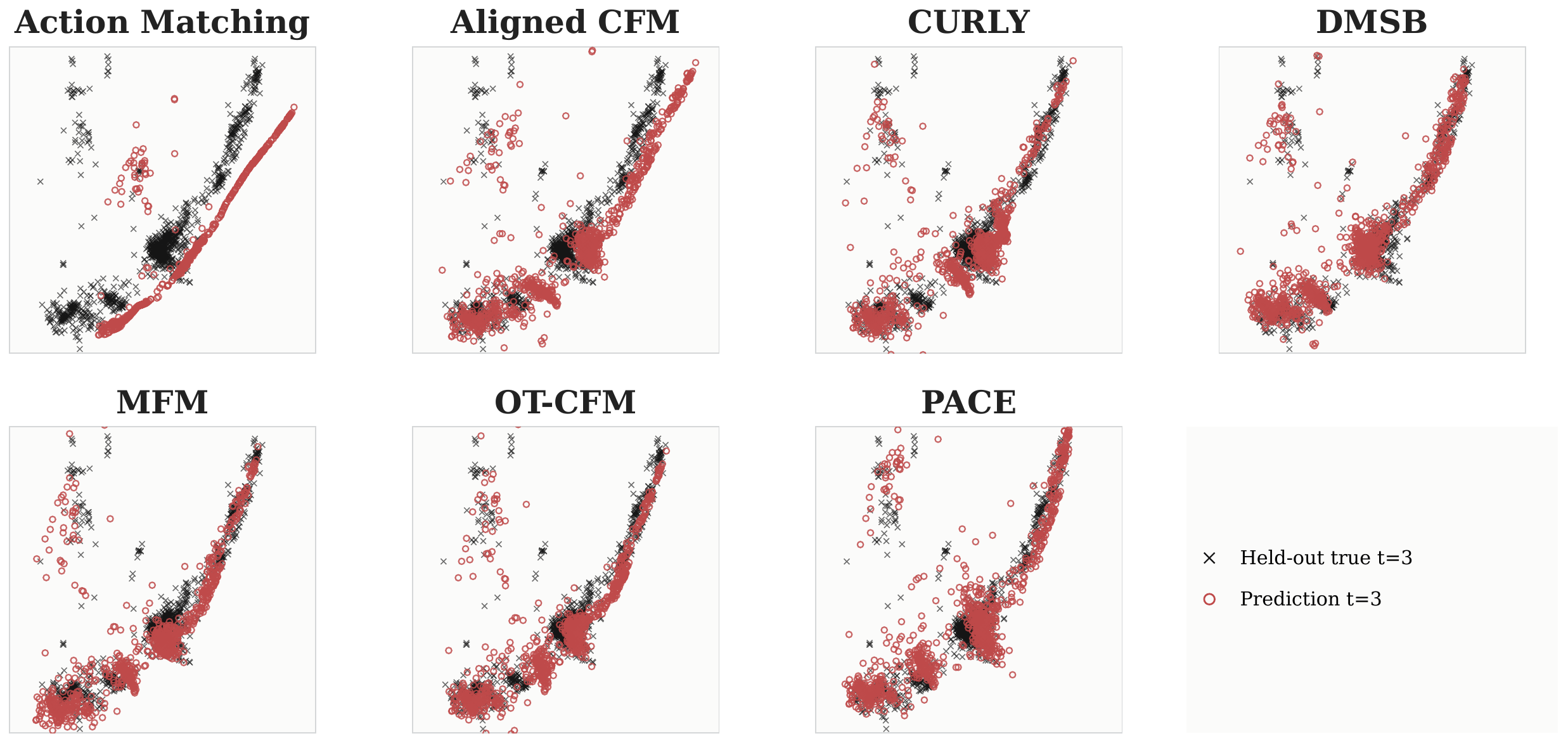}
    \caption{Qualitative Stage-2 results on EB PHATE~\citep{moon2019visualizing}. The top panel shows baseline rollouts; the bottom panel compares predicted and held-out point clouds.}
    \label{fig:app:stage2-eb}
\end{figure}

\begin{figure}[p]
    \centering
    \includegraphics[width=\linewidth]{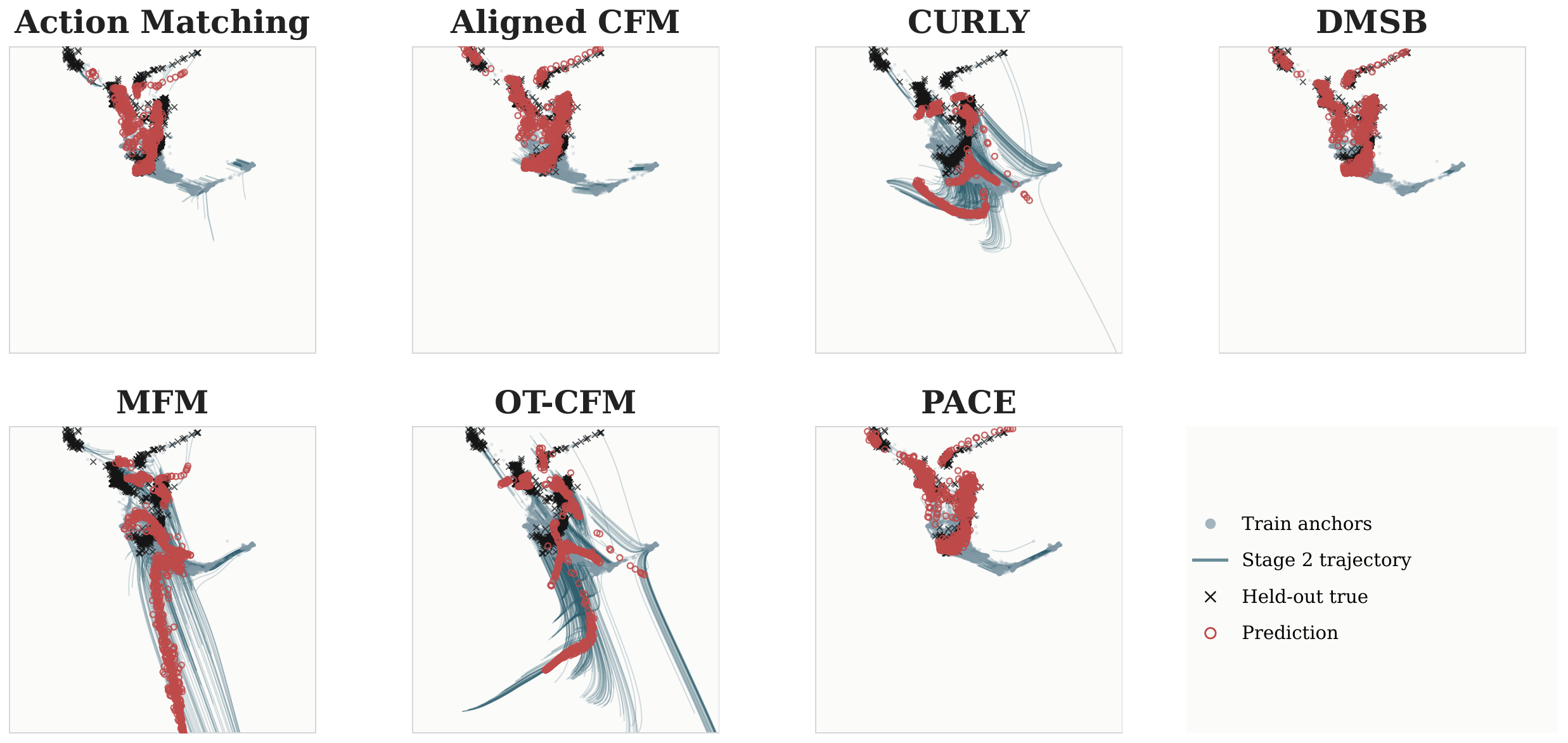}
    \vspace{0.4em}
    \includegraphics[width=\linewidth]{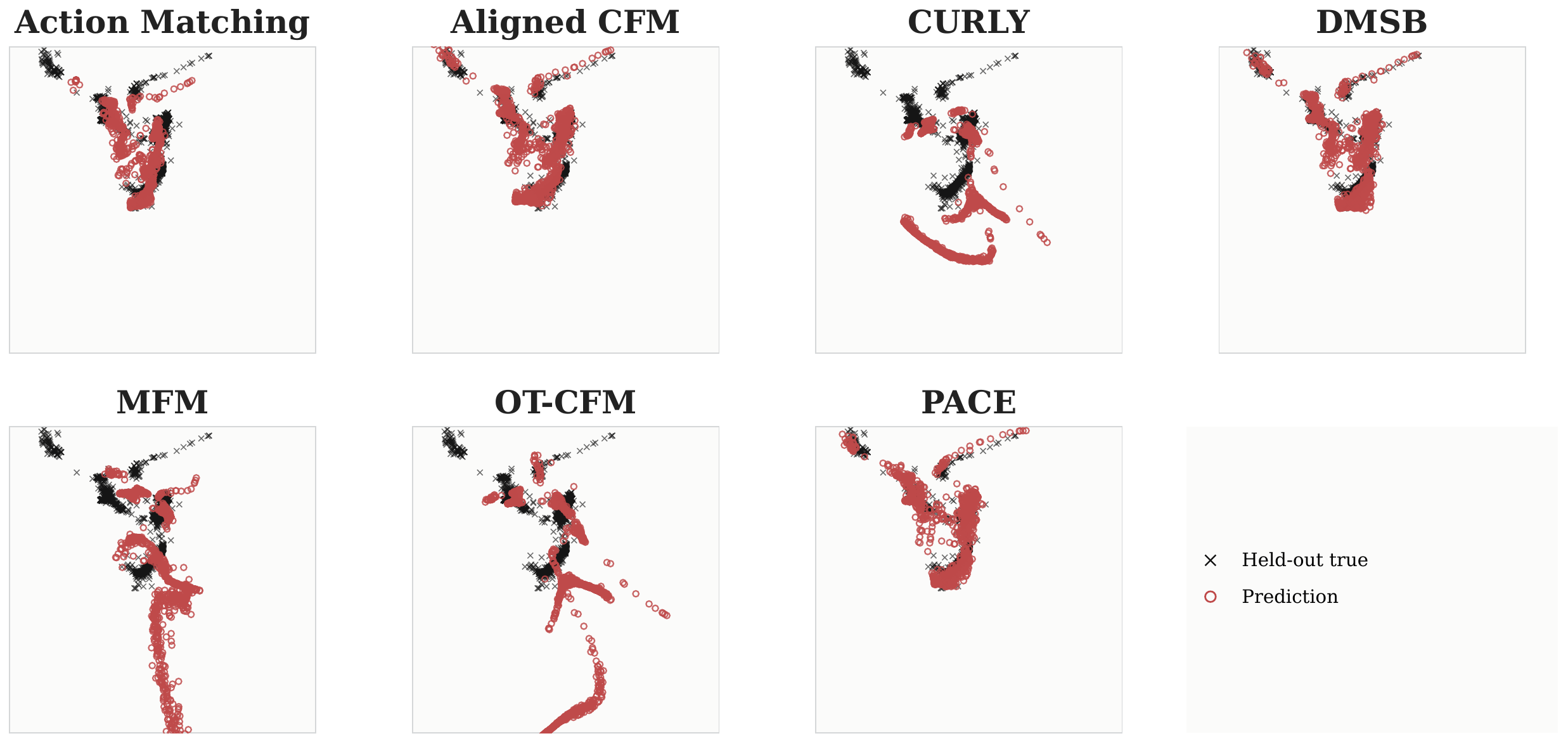}
    \caption{Qualitative Stage-2 results on Schiebinger2019~\citep{schiebinger2019optimal}. The top panel shows baseline rollouts; the bottom panel compares predicted and held-out point clouds.}
    \label{fig:app:stage2-schiebinger}
\end{figure}

\begin{figure}[p]
    \centering
    \includegraphics[width=\linewidth]{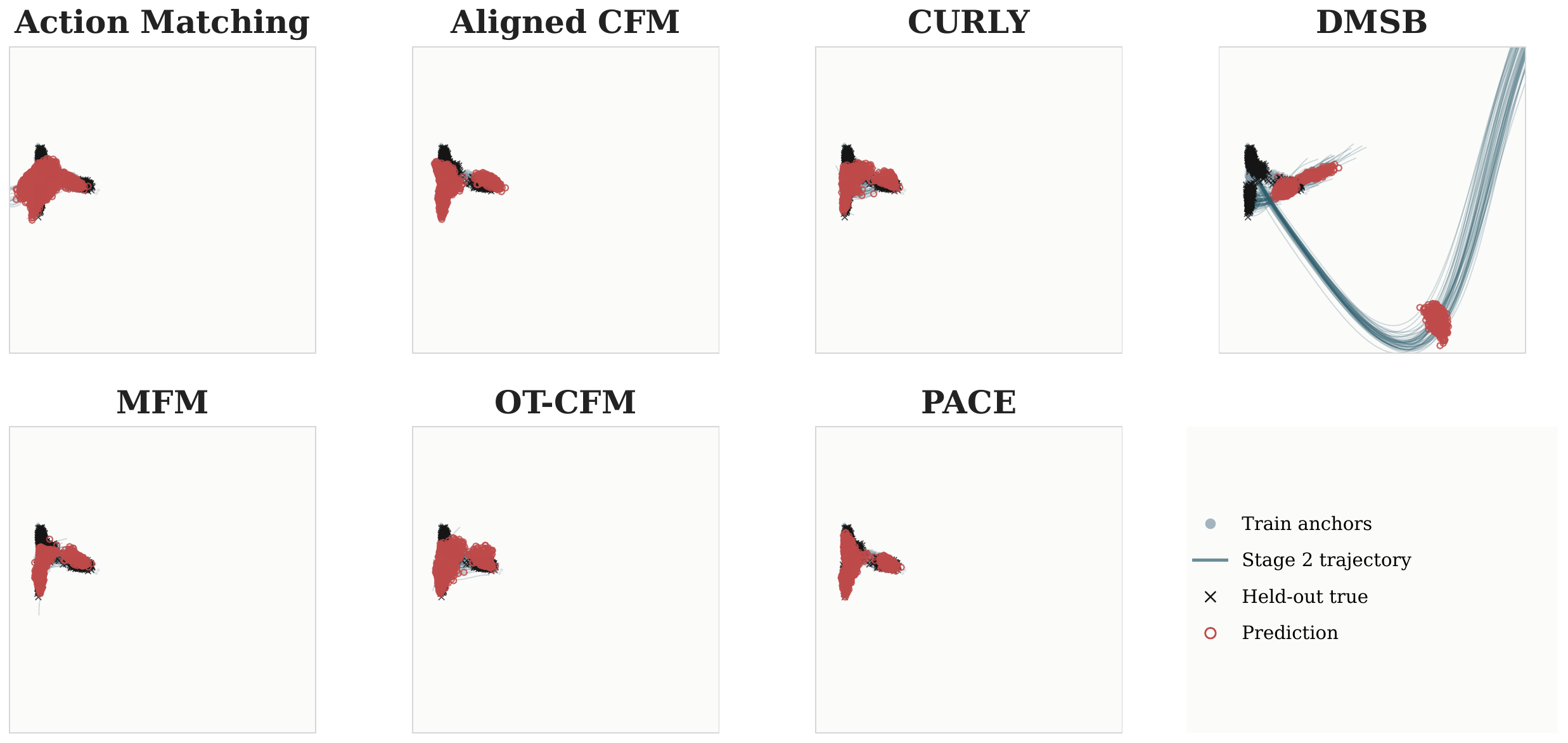}
    \vspace{0.4em}
    \includegraphics[width=\linewidth]{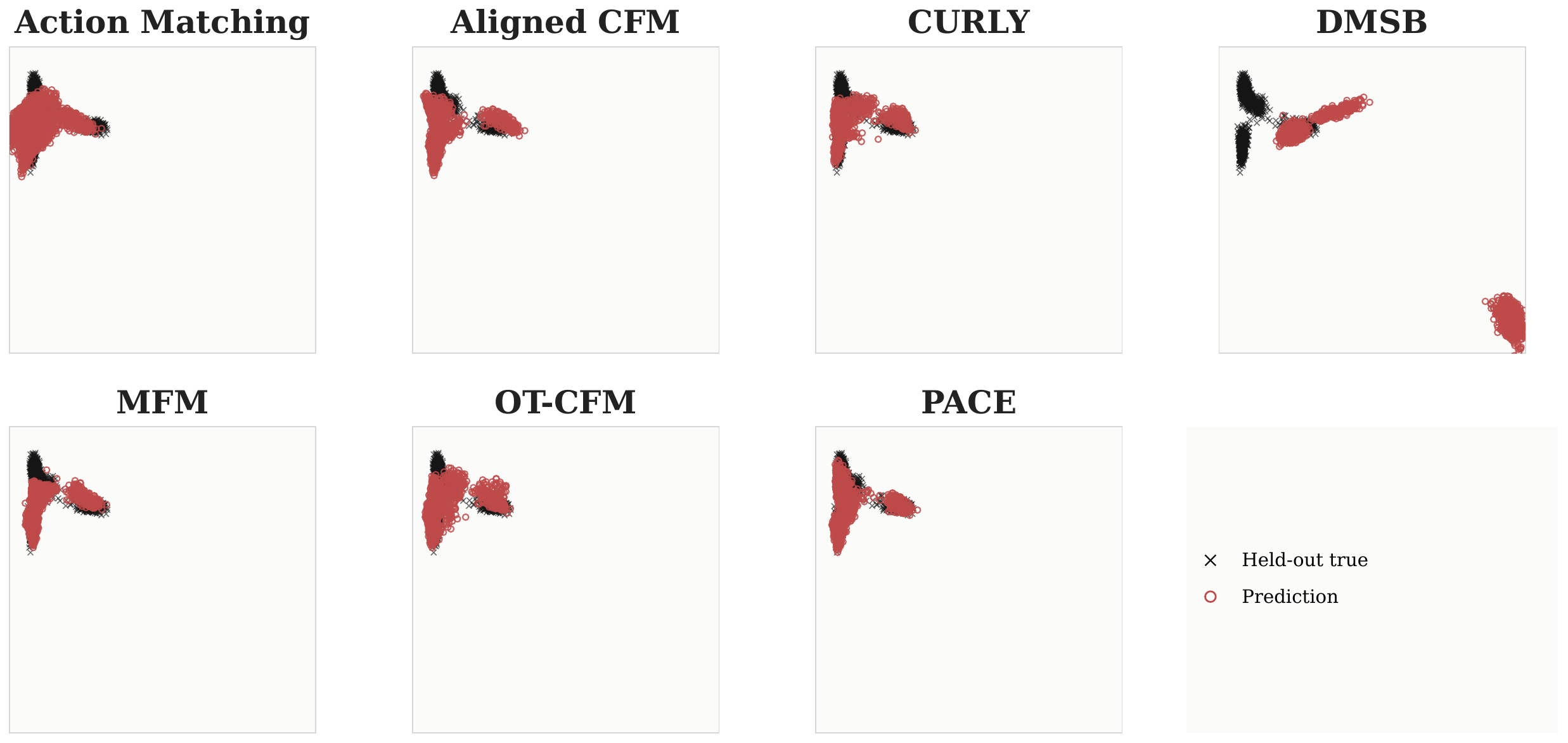}
    \caption{Qualitative Stage-2 results on iPSC-Liu~\citep{liu2020reprogramming}. The top panel shows baseline rollouts; the bottom panel compares predicted and held-out point clouds.}
    \label{fig:app:stage2-ipsc}
\end{figure}

\clearpage
\section{Baseline Details}
\label{sec:app:baselines}

The result tables include only methods that were run for the corresponding experiment. All deterministic flow baselines use the same dataloaders, train/test labels, whitening rules, and held-out ODE rollout evaluator described in Appendix~\ref{sec:app:benchmarks}. Unless an experiment-specific config overrides it, flow networks are MLP velocity fields trained with Adam/AdamW-style optimizers. The formulas below describe the implemented objectives under a segment $[t_k,t_{k+1}]$, with $\Delta t=t_{k+1}-t_k$ and normalized segment time $s=(t-t_k)/\Delta t\in[0,1]$.

\paragraph{Action Matching.}
Action Matching~\citep{neklyudov2023action} learns a scalar action network $S_\theta(t,x)$. For endpoint samples $(x_0,x_1)$ and an interpolated point $x_t=(1-s)x_0+s x_1$, the implemented segment loss has the form
\begin{equation}
\begin{aligned}
\mathcal{L}_{\mathrm{AM}}
=\mathbb{E}\big[&w(t_k)S_\theta(t_k,x_0)-w(t_{k+1})S_\theta(t_{k+1},x_1) \\
&+\Delta t\{w(t)(\partial_t S_\theta(t,x_t)+\tfrac{1}{2}\|\nabla_x S_\theta(t,x_t)\|^2)
+S_\theta(t,x_t)\partial_t w(t)\}\big].
\end{aligned}
\end{equation}
At test time, the velocity is $v_\theta(t,x)=\nabla_x S_\theta(t,x)$ and predictions are generated by solving $\dot{x}_t=v_\theta(t,x_t)$.

\paragraph{OT-CFM.}
OT-CFM~\citep{tong2023conditional} pairs mini-batches by an optimal-transport plan
\begin{equation}
  \pi_k^*\in\arg\min_{\pi\in\Pi(\hat\rho_k,\hat\rho_{k+1})}
  \mathbb{E}_{(x_0,x_1)\sim\pi}\|x_0-x_1\|^2.
\end{equation}
For paired endpoints, it uses the straight conditional path
\begin{equation}
  x_t=(1-s)x_0+s x_1,
  \qquad
  u_t=\frac{x_1-x_0}{\Delta t},
\end{equation}
and trains a velocity network with the flow-matching objective $\mathbb{E}\|v_\theta(t,x_t)-u_t\|^2$. The implementation uses the exact minibatch OT matcher from \texttt{torchcfm} and sets the conditional path noise to zero.

\paragraph{MFM.}
Metric Flow Matching~\citep{kapusniak2024metric} uses a two-stage pipeline. It first trains a geodesic correction network $\psi_\eta$ under a data-induced metric, then trains a velocity field using the corrected interpolant
\begin{equation}
  \mu_t=(1-s)x_0+s x_1+\gamma(t)\psi_\eta(x_0,x_1,t),
\end{equation}
where $\gamma(t)$ vanishes at the segment endpoints. The flow-matching target is the time derivative
\begin{equation}
  u_t=\frac{x_1-x_0}{\Delta t}+\dot{\gamma}(t)\psi_\eta(x_0,x_1,t)
  +\gamma(t)\partial_t\psi_\eta(x_0,x_1,t),
\end{equation}
with the last term present when the correction network depends explicitly on time. Our runs use the LAND-style data-manifold metric configuration.

\paragraph{CURLY.}
CURLY~\citep{petrovic2026curly} uses a related corrected interpolant, but with the normalized-time modulation $s(1-s)$:
\begin{equation}
  \mu_t=(1-s)x_0+s x_1+s(1-s)\psi_\eta(x_0,x_1,s).
\end{equation}
The target velocity used for the second-stage flow model is
\begin{equation}
  u_t=\frac{1}{\Delta t}\left[(x_1-x_0)+(1-2s)\psi_\eta
  +s(1-s)\partial_s\psi_\eta\right].
\end{equation}
The first stage fits $\psi_\eta$ using pseudo-velocity supervision estimated from the observed temporal point clouds.

\paragraph{Aligned CFM.}
Aligned CFM follows the adversarially learned interpolant approach of ALI-CFM~\citep{kviman2026multimarginal}. It trains a global interpolant $I_\eta(x_0,x_1,t)$ from the first to the last training snapshot. Intermediate marginals of $I_\eta$ are aligned with observed training snapshots by an adversarial objective of the schematic form
\begin{equation}
  \min_\eta\max_D\sum_{\ell}
  \left[\mathbb{E}_{x\sim\hat\rho_\ell}\log D_\ell(x)
  +\mathbb{E}_{x_0,x_1}\log\{1-D_\ell(I_\eta(x_0,x_1,t_\ell))\}\right].
\end{equation}
After this stage, the interpolant is frozen and a velocity model is trained with
\begin{equation}
  \mathcal{L}_{\mathrm{ALI-CFM}}
  =\mathbb{E}\|v_\theta(t,I_\eta(x_0,x_1,t))-\partial_t I_\eta(x_0,x_1,t)\|^2.
\end{equation}

\paragraph{DMSB.}
DMSB~\citep{chen2023deep} is a stochastic bridge baseline in joint position-velocity state space $z_t=(x_t,v_t)$. The implementation uses a momentum SDE discretization with learned control $a_\theta(z_t,t)$:
\begin{equation}
  z_{n+1}=z_n+\begin{bmatrix}v_n\\0\end{bmatrix}\Delta t
  +\sigma\left(\begin{bmatrix}0\\a_\theta(z_n,t_n)\end{bmatrix}\Delta t
  +\begin{bmatrix}0\\\epsilon_n\end{bmatrix}\right),
  \qquad \epsilon_n\sim\mathcal{N}(0,\Delta t I).
\end{equation}
Forward and backward policies are trained on the temporal grid, and the learned dynamics are rolled out to predict the held-out marginal distributions. DMSB appears only in the tables where this baseline was run.

\section{PACE Implementation Details}
\label{sec:app:pace-implementation}

PACE is implemented as a two-stage procedure. Stage 1 learns endpoint-preserving bridges and refines the cross-time matching; Stage 2 freezes the learned bridge correction and distills it into a global ODE velocity field for rollout.
The main-text benchmark experiments were run on a single NVIDIA A100 80GB GPU. The additional wall-clock measurements in Appendix~\ref{app:efficiency} were performed on a single NVIDIA H800.

\paragraph{Stage 1 anchor construction.}
For each experiment, the selected training frames are first grouped by time label. PACE stacks the training frames into an anchor tensor of shape $K\times N\times d$, where $K$ is the number of training labels and $N$ is the minimum selected cell count across training time points. If time points have unequal counts, each frame is trimmed to this common $N$ for the Stage 1 matching problem. This produces one adjacent matching problem per training interval.

\paragraph{Initial and refined matchings.}
The initial matching between adjacent anchors uses a cost combining normalized squared Euclidean distance with the local normal-direction penalty.
After a fixed number of bridge-training epochs, PACE recomputes the matching by evaluating the learned path-action cost and solving the induced OT problem.
We use the Python Optimal Transport (POT) package for OT computations.
Rematching is repeated at a fixed experiment-specific interval (we usually choose 10 or 20), so bridge learning and correspondence estimation alternate throughout training.

\paragraph{Bridge and flow distillation.}
For endpoints $(x_0,x_1)$ and local time $s$, the bridge has the endpoint-preserving form
\begin{equation}
  \gamma_\theta(x_0,x_1,s)
  =
  (1-s)x_0+s x_1+s(1-s)\psi_\theta(x_0,x_1,s).
\end{equation}
The bridge velocity $\partial_s\gamma_\theta$ is obtained by automatic differentiation in 2D and by a Jacobian-vector product implementation in higher dimensions. After Stage 1, $\psi_\theta$ is frozen. Stage 2 samples paired endpoints from the Stage 1 matchings and trains a global velocity network $v_\phi(t,x)$ by regressing to the bridge velocity along generated points. Held-out predictions in the tables are produced by rolling out this Stage 2 velocity field.

\section{Proof of Proposition~\ref{prop:illposed} (Ill-posedness)}
\label{app:illposed}

\begin{proof}[Proof of Proposition~\ref{prop:illposed}]
For one interval $[t_k,t_{k+1}]$, many couplings $\pi_k\in\Pi(\hat\rho_k,\hat\rho_{k+1})$ can share the same endpoint marginals $\hat\rho_k$ and $\hat\rho_{k+1}$. For any such coupling, each coupled endpoint pair $(x,y)$ can be connected by infinitely many $C^1$ paths $\gamma:[0,1]\to\mathbb{R}^d$ with $\gamma(0)=x$ and $\gamma(1)=y$. Therefore, the observed snapshots determine neither a unique coupling nor a unique intermediate trajectory.
\end{proof}

\section{Local Metric and Correspondence}
\label{app:metric-correspondence}

The metric used by PACE changes the matching criterion from endpoint proximity to path plausibility. In the locally constant idealization, suppose
\[
  G=I+\alpha P_N,\qquad \alpha>0,
\]
where $P_N$ is the orthogonal projector onto the local normal subspace. For a straight candidate displacement $d=y-x$, the metric action is
\begin{equation}
  d^\top Gd
  =
  \|d\|^2+\alpha\|P_Nd\|^2.
\end{equation}
Thus, two candidate matches with the same Euclidean distance are not equivalent: the one whose displacement is more tangent-aligned has lower action. The implemented initial matching uses this principle through a normal-projection penalty, while the refined matching applies the same idea to the learned curved bridge by averaging
\begin{equation}
  \dot\gamma_\theta(s)^\top
  G(\gamma_\theta(s),t(s))
  \dot\gamma_\theta(s)
\end{equation}
over a finite probe grid. This is the practical reason PACE can prefer a slightly longer but tangent-compatible path over a shorter chord that cuts across the inferred developmental geometry.

\section{Monotonicity of Ideal Alternating Updates}
\label{app:optimization-monotonicity}

The alternating bridge-coupling updates can be viewed as block-coordinate descent for the discretized metric-action objective. For each adjacent time interval \(k\), let
\[
\Pi_k=\Pi(\hat\rho_k,\hat\rho_{k+1})
\]
denote the feasible set of empirical couplings, and define the finite-dimensional objective
\begin{equation}
\label{eq:app-joint-action}
\mathcal{J}(\theta,\{\pi_k\}_{k=0}^{K-1})
=
\sum_{k=0}^{K-1}
\sum_{i,j}
\pi_{ij}^{(k)}
c_{ij}^{\mathrm{path}}(\theta),
\end{equation}
where \(c_{ij}^{\mathrm{path}}(\theta)\) is the discretized path-action cost in Eq.~\eqref{eq:path-cost}. This objective is the full-batch version of the metric-action term optimized during the bridge update, using the same path-action quadrature as the coupling update.

\begin{proposition}[Monotonicity under exact block updates]
\label{prop:ideal-bcd-monotonicity}
Assume that, at iteration \(r\), the bridge update computes a global minimizer
\[
\theta^{r+1}\in
\arg\min_{\theta}
\mathcal{J}(\theta,\{\pi_k^r\}_{k=0}^{K-1}),
\]
and the coupling update computes, for each \(k\), a global minimizer
\[
\pi_k^{r+1}
\in
\arg\min_{\pi_k\in\Pi_k}
\sum_{i,j}
\pi_{ij}^{(k)}
c_{ij}^{\mathrm{path}}(\theta^{r+1}).
\]
Then the joint metric-action objective is non-increasing:
\begin{equation}
\mathcal{J}(\theta^{r+1},\{\pi_k^{r+1}\}_{k=0}^{K-1})
\le
\mathcal{J}(\theta^{r},\{\pi_k^{r}\}_{k=0}^{K-1}).
\end{equation}
Moreover, the decrease is strict whenever either the bridge block or at least one coupling block achieves a strict improvement over the previous iterate.
\end{proposition}

\begin{proof}
By the optimality of the bridge update with \(\{\pi_k^r\}\) fixed,
\[
\mathcal{J}(\theta^{r+1},\{\pi_k^r\}_{k=0}^{K-1})
\le
\mathcal{J}(\theta^{r},\{\pi_k^r\}_{k=0}^{K-1}).
\]
By the optimality of each coupling update with \(\theta^{r+1}\) fixed,
\[
\sum_{i,j}\pi_{ij}^{(k),r+1}c_{ij}^{\mathrm{path}}(\theta^{r+1})
\le
\sum_{i,j}\pi_{ij}^{(k),r}c_{ij}^{\mathrm{path}}(\theta^{r+1})
\quad
\text{for every }k.
\]
Summing these inequalities over \(k\) gives
\[
\mathcal{J}(\theta^{r+1},\{\pi_k^{r+1}\}_{k=0}^{K-1})
\le
\mathcal{J}(\theta^{r+1},\{\pi_k^r\}_{k=0}^{K-1}).
\]
Combining the two inequalities proves monotonicity. If either inequality is strict, the combined decrease is strict.
\end{proof}

Because \(G_k(x,t)=I+\alpha C_N^{(k)}(x,t)\) is positive definite, every path-action cost is nonnegative; therefore the sequence of ideal objective values is bounded below by zero and hence converges as a sequence of numbers. This does not imply convergence to a global optimum, nor does it imply uniqueness of the bridge or coupling. Equality can occur when a block update returns an equivalent minimizer or when the current block is already optimal.

\paragraph{Relation to the implemented algorithm.}
The proposition applies to the ideal metric-action problem in Eq.~\eqref{eq:app-joint-action}. The implemented PACE training approximates this scheme with stochastic minibatches, a neural bridge optimized by finite gradient steps, periodic rather than continuous rematching, and optional stabilizing regularizers. These choices are used for scalability and numerical stability, but they do not provide a per-gradient-step strict decrease guarantee for the realized neural training trajectory. The monotonic result should therefore be read as the optimization principle behind the alternating updates, while the empirical sections evaluate the behavior of the practical stochastic implementation.

\section{Stage 1 Regularizers}
\label{app:regularizers}

PACE Stage 1 optimizes a metric action together with an optional regularization term:
\begin{equation}
  \mathcal{L}_{\mathrm{bridge}}
  =
  \lambda_{\mathrm{metric}}\mathcal{L}_{\mathrm{metric}}
  +
  \lambda_{\mathrm{reg}}\mathcal{L}_{\mathrm{reg}}.
\end{equation}
The metric-action term is the discretized version of Eq.~\eqref{eq:metric-loss}. The regularization term can include the implementation-level stabilizers below; in our experiments, this corresponds to a weighted combination
\[
  \lambda_{\mathrm{reg}}\mathcal{L}_{\mathrm{reg}}
  =
  \lambda_{\mathrm{coh}}\mathcal{L}_{\mathrm{coh}}
  +
  \lambda_{\mathrm{orth}}\mathcal{L}_{\mathrm{orth}}.
\]

\paragraph{Cross-segment velocity coherence.}
Let $(x_a,t_a,v_a)$ and $(x_b,t_b,v_b)$ denote generated bridge points and velocities from different adjacent segments. PACE constructs Gaussian space-time weights
\begin{equation}
  W_{ab}
  =
  \exp\!\left(-\frac{\|x_a-x_b\|^2}{\sigma_{x,a}\sigma_{x,b}}\right)
  \exp\!\left(-\frac{|t_a-t_b|^2}{\sigma_{t,a}\sigma_{t,b}}\right)
  \mathbf{1}\{\mathrm{seg}(a)\neq\mathrm{seg}(b)\},
\end{equation}
where the bandwidths are estimated from local projector variation. The coherence loss penalizes nearby generated points from different segments when their normalized velocities disagree:
\begin{equation}
  \mathcal{L}_{\mathrm{coh}}
  =
  \frac{\sum_{a,b}W_{ab}\left(1-\hat v_a^\top \hat v_b\right)}
  {\sum_{a,b}W_{ab}+\varepsilon},
  \qquad
  \hat v_a=\frac{v_a}{\|v_a\|+\varepsilon}.
\end{equation}

\paragraph{Normal-motion suppression.}
For generated bridge points, PACE estimates a local normal projector and penalizes the normal component of the generated velocity:
\begin{equation}
  \mathcal{L}_{\mathrm{orth}}
  =
  \frac{1}{|\mathcal{G}|}
  \sum_{(x,t,v)\in\mathcal{G}}
  \|P_N(x,t)v\|^2.
\end{equation}
In 2D this reduces to a squared dot product with the locally estimated normal vector. In higher dimensions it uses the full normal-space projector $P_N=I-P_T$.

\section{High-Dimensional Concentration Effects}
\label{sec:app:concentration}

The high-dimensional experiments in Section~\ref{sec:experiment} operate in a regime where Euclidean probability mass and pairwise costs are strongly concentrated. This section explains the diagnostics used in Figure~\ref{fig:norm-time-concentration-vs-dim} and why concentration affects OT, MMD, and nearest-neighbor-based trajectory objectives.

\paragraph{Norm concentration.}
For \(X\sim\mathcal{N}(0,I_d)\), the map \(f(x)=\|x\|_2\) is \(1\)-Lipschitz. The Gaussian concentration inequality therefore gives
\begin{equation}
  \Pr\!\left(\left|\|X\|_2-\mathbb{E}\|X\|_2\right|\ge t\right)
  \le 2\exp(-t^2/2),
\end{equation}
and \(\mathbb{E}\|X\|_2\asymp \sqrt{d}\) implies
\begin{equation}
  \frac{\operatorname{Std}(\|X\|_2)}{\mathbb{E}\|X\|_2}
  = O(d^{-1/2}).
\end{equation}
This is the standard Lipschitz concentration phenomenon~\citep{ledoux2001concentration}: after whitening or PCA scaling, most points are pushed toward a thin shell rather than spreading across many radial scales.

\paragraph{Direction and nearest-neighbor concentration.}
For independent isotropic sub-Gaussian vectors \(X,Y\in\mathbb{R}^d\), normalized directions satisfy
\begin{equation}
  \left\langle \frac{X}{\|X\|_2},\frac{Y}{\|Y\|_2}\right\rangle
  = O_p(d^{-1/2}),
\end{equation}
so random directions become nearly orthogonal as \(d\) grows~\citep{vershynin2018high}. A related nearest-neighbor result states that if the relative variance of distances vanishes,
\begin{equation}
  \frac{\operatorname{Var}(\|Q-X_1\|_2)}
       {\mathbb{E}[\|Q-X_1\|_2]^2}
  \to 0,
\end{equation}
then
\begin{equation}
  \Pr\!\left(
  \frac{D_{\max}^{(d)}-D_{\min}^{(d)}}{D_{\min}^{(d)}}\le \epsilon
  \right)\to 1
  \qquad \text{for every }\epsilon>0,
\end{equation}
where \(D_{\min}^{(d)}\) and \(D_{\max}^{(d)}\) are the nearest and farthest distances from the query point~\citep{beyer1999nearest}. Thus, when distance contrast collapses, the nearest neighbor and farthest neighbor become less distinguishable in relative terms.

\paragraph{Effect on OT, MMD, and flow-matching losses.}
If two approximately isotropic samples \(X,Y\in\mathbb{R}^d\) have coordinate-wise fluctuations with comparable scale, then
\begin{equation}
  \|X-Y\|_2^2 = \sum_{j=1}^{d} (X_j-Y_j)^2
\end{equation}
is a sum of many coordinate-level contributions. Under standard independence or weak-dependence assumptions, a law-of-large-numbers or sub-Gaussian concentration argument implies that \(\|X-Y\|_2^2/d\) concentrates around its mean, with relative fluctuations that typically scale as \(O(d^{-1/2})\). For transport objectives,
\begin{equation}
  W_2^2(\mu,\nu)
  = \inf_{\pi\in\Pi(\mu,\nu)}
    \mathbb{E}_{(x,y)\sim\pi}\|x-y\|_2^2,
\end{equation}
so concentration directly reduces the dynamic range of the cost matrix used by OT and flow-matching couplings. For kernel metrics,
\begin{equation}
  \operatorname{MMD}^2(\mu,\nu)
  =
  \mathbb{E}\big[k(x,x')+k(y,y')-2k(x,y)\big],
\end{equation}
and a distance-based kernel such as \(k(x,y)=\exp(-\|x-y\|_2^2/(2\sigma^2))\) also loses contrast when most \(\|x-y\|_2^2\) values occupy a narrow interval. Consequently, different predicted distributions can receive similar numerical losses even when their local geometry differs. This is the metric-degeneration issue encountered by high-dimensional distance-based trajectory methods, including metric and OT flow-matching objectives~\citep{kapusniak2024metric}.

\paragraph{Reference thresholds in Figure~\ref{fig:norm-time-concentration-vs-dim}.}
The first diagnostic is
\begin{equation}
  \operatorname{CV}_{\rm norm}
  =
  \frac{\operatorname{Std}(\|X\|_2)}{\mathbb{E}\|X\|_2}.
\end{equation}
The horizontal reference level \(0.3\) is a practical warning threshold: below it, one standard deviation of radial variation is less than \(30\%\) of the mean radius, so most cells lie in a relatively thin shell and pairwise Euclidean costs are dominated by small angular or local fluctuations. This threshold is not a theorem-specific cutoff; it is a scale marker chosen to make the \(O(d^{-1/2})\) collapse visible on the empirical PCA representations.

The second diagnostic compares time separation with within-time dispersion. For snapshot centers \(c_i,c_j\), define
\begin{equation}
  R_{\rm time}
  =
  \operatorname{median}_{i<j}
  \frac{\|c_i-c_j\|_2}{\bar r_{ij}},
  \qquad
  \bar r_{ij}
  =
  \frac{1}{2}
  \left(
  \mathbb{E}_{X\sim \mu_i}\|X-c_i\|_2
  +
  \mathbb{E}_{Y\sim \mu_j}\|Y-c_j\|_2
  \right).
\end{equation}
The reference level \(1.0\) marks the point where the median displacement between time-point centers is comparable to the typical within-time radius. When \(R_{\rm time}\le 1\), the Euclidean shift between snapshots is no larger than the spread of the snapshots themselves, so time labels are difficult to separate using raw pairwise distances alone. Figure~\ref{fig:norm-time-concentration-vs-dim} summarizes these diagnostics in the main text. The appendix diagnostics in Figures~\ref{fig:app:ipsc_dim} and~\ref{fig:app:ipsc} show the corresponding norm-concentration behavior for iPSC-Liu, OP-Cite, and OP-Multi.

\paragraph{High-dimensional benchmark results.}
Table~\ref{tab:ipsc_liu_dim10_50} reports the per-timepoint reconstruction metrics for iPSC-Liu~\citep{liu2020reprogramming} in 10D and 50D PCA representations, evaluated on holdout time points $t\in\{4,8,12,16\}$.
Table~\ref{tab:op_100d} gives the corresponding results for the multimodal OP-Cite and OP-Multi datasets~\citep{lance2022multimodal} in 100D.
These tables complement the time-averaged summary in Table~\ref{tab:ipsc_liu_dim10_50_mean} and the concentration diagnostics in Figure~\ref{fig:app:ipsc_dim}.

\begin{table}[t]
\centering
\caption{iPSC-Liu~\citep{liu2020reprogramming} (10D/50D) per-timepoint results across holdouts $t\in\{4,8,12,16\}$.}
\label{tab:ipsc_liu_dim10_50}
\small
\setlength{\tabcolsep}{3pt}
\resizebox{1\linewidth}{!}{
\begin{tabular}{llcccccccccccc}
\toprule
Dim & Method & \multicolumn{4}{c}{MMD $\downarrow$} & \multicolumn{4}{c}{$\mathcal{W}_1\downarrow$} & \multicolumn{4}{c}{$\mathcal{W}_2\downarrow$} \\
\cmidrule(lr){3-6} \cmidrule(lr){7-10} \cmidrule(lr){11-14}
 & & $t=4$ & $t=8$ & $t=12$ & $t=16$ & $t=4$ & $t=8$ & $t=12$ & $t=16$ & $t=4$ & $t=8$ & $t=12$ & $t=16$ \\
\midrule
10D & Action Matching & 0.6504 & 0.6959 & 0.5132 & 0.3626 & 3.8458 & 4.1608 & 3.1014 & 4.0123 & 4.2076 & 5.9396 & 3.4279 & 3.9222 \\
 & Aligned CFM & 0.6245 & \underline{0.5936} & \underline{0.4051} & \textbf{0.2345} & 3.7937 & \underline{3.4433} & \underline{2.5372} & \textbf{1.9877} & 4.2084 & \underline{4.0761} & \underline{2.9973} & \textbf{2.2482} \\
 & CURLY & \underline{0.5461} & 0.6054 & 0.4936 & 0.5024 & 3.8694 & 3.9773 & 3.1387 & 3.4660 & 4.3490 & 5.5607 & 3.4974 & 3.4336 \\
 & DMSB & 0.9250 & 0.7057 & 0.7448 & 0.7098 & 10.5837 & 4.1953 & 5.5040 & 4.9522 & 10.6553 & 4.5341 & 5.6287 & 5.0190 \\
 & MFM & 0.5864 & 0.6674 & 0.5033 & 0.3521 & \underline{3.5972} & 4.0239 & 3.2515 & 2.7539 & \underline{4.0205} & 5.0654 & 3.5968 & 2.6699 \\
 & OT-CFM & 0.6039 & 0.6339 & 0.4951 & 0.3251 & 3.6543 & 4.0040 & 3.2932 & 2.5884 & 4.0624 & 5.5807 & 3.6325 & \underline{2.5030} \\
\cmidrule(lr){2-14}
 & \textbf{\model (ours)} & \textbf{0.5244} & \textbf{0.5153} & \textbf{0.3113} & \underline{0.2871} & \textbf{3.5918} & \textbf{3.2507} & \textbf{2.3185} & \underline{2.5294} & \textbf{3.7485} & \textbf{3.5755} & \textbf{2.7868} & 2.6406 \\
\midrule
50D & Action Matching & \textbf{0.3517} & 0.3506 & 0.3278 & 0.2197 & 7.2261 & 8.4554 & 7.9593 & 6.8096 & 7.8273 & 10.3927 & 8.5704 & 7.1567 \\
 & Aligned CFM & \underline{0.3519} & \textbf{0.2917} & 0.2397 & \underline{0.1663} & \underline{7.0562} & \underline{7.5723} & \underline{7.5947} & 7.0302 & \underline{7.6855} & \underline{8.2791} & \underline{8.2483} & 8.0080 \\
 & CURLY & 0.3824 & 0.4189 & \underline{0.2235} & 0.2327 & 8.2766 & 9.7010 & 7.6085 & 6.8715 & 8.8476 & 11.2946 & 8.2990 & 7.1972 \\
 & DMSB & 0.7593 & 0.5730 & 0.4787 & 0.4441 & 14.5899 & 10.4588 & 9.7531 & 8.9456 & 14.7121 & 10.7933 & 10.1192 & 9.3447 \\
 & MFM & 0.3539 & 0.3467 & 0.2679 & 0.1728 & \textbf{7.0197} & 9.2505 & 8.5746 & \underline{6.3246} & \textbf{7.6665} & 10.5394 & 9.2502 & \textbf{6.7254} \\
 & OT-CFM & 0.3530 & 0.3458 & 0.2802 & 0.1803 & 7.1278 & 9.0608 & 8.7395 & 6.4288 & 7.7653 & 10.4423 & 9.3687 & 6.8054 \\
\cmidrule(lr){2-14}
 & \textbf{\model (ours)} & 0.3609 & \underline{0.3004} & \textbf{0.2193} & \textbf{0.1396} & 7.0656 & \textbf{6.8194} & \textbf{6.8718} & \textbf{6.2581} & 7.6896 & \textbf{7.0798} & \textbf{7.5736} & \underline{6.7914} \\
\bottomrule
\end{tabular}
}
\end{table}

\begin{table}[t]
\centering
\caption{OP-Cite and OP-Multi~\citep{lance2022multimodal} (100D) results.}
\label{tab:op_100d}
\small
\setlength{\tabcolsep}{4pt}
\resizebox{0.75\linewidth}{!}{
\begin{tabular}{lcccccc}
\toprule
Method & \multicolumn{3}{c}{\textsc{OP-Cite} (100D)} & \multicolumn{3}{c}{\textsc{OP-Multi} (100D)} \\
\cmidrule(lr){2-4} \cmidrule(lr){5-7}
 & MMD $\downarrow$ & $\mathcal{W}_1\downarrow$ & $\mathcal{W}_2\downarrow$ & MMD $\downarrow$ & $\mathcal{W}_1\downarrow$ & $\mathcal{W}_2\downarrow$ \\
\midrule
Action Matching & 0.1753 & \textbf{9.9388} & \textbf{10.0448} & \textbf{0.1622} & 10.7711 & 10.8160 \\
Aligned CFM & \underline{0.1434} & 10.7009 & 10.8094 & 0.1720 & 10.8330 & 10.8826 \\
CURLY & 0.2529 & 11.5163 & 11.6015 & 0.2786 & 12.0015 & 12.0440 \\
MFM & 0.1435 & \underline{10.4301} & \underline{10.5409} & 0.1699 & 10.6274 & 10.6702 \\
OT-CFM & 0.1448 & 10.4447 & 10.5491 & \underline{0.1662} & \underline{10.6128} & \underline{10.6564} \\
\midrule
\textbf{\model (ours)} & \textbf{0.1409} & 10.4574 & 10.5605 & 0.1693 & \textbf{10.5824} & \textbf{10.6293} \\
\bottomrule
\end{tabular}
}
\end{table}

\begin{figure}[hbtp]
    \centering
    \includegraphics[width=1.0\linewidth]{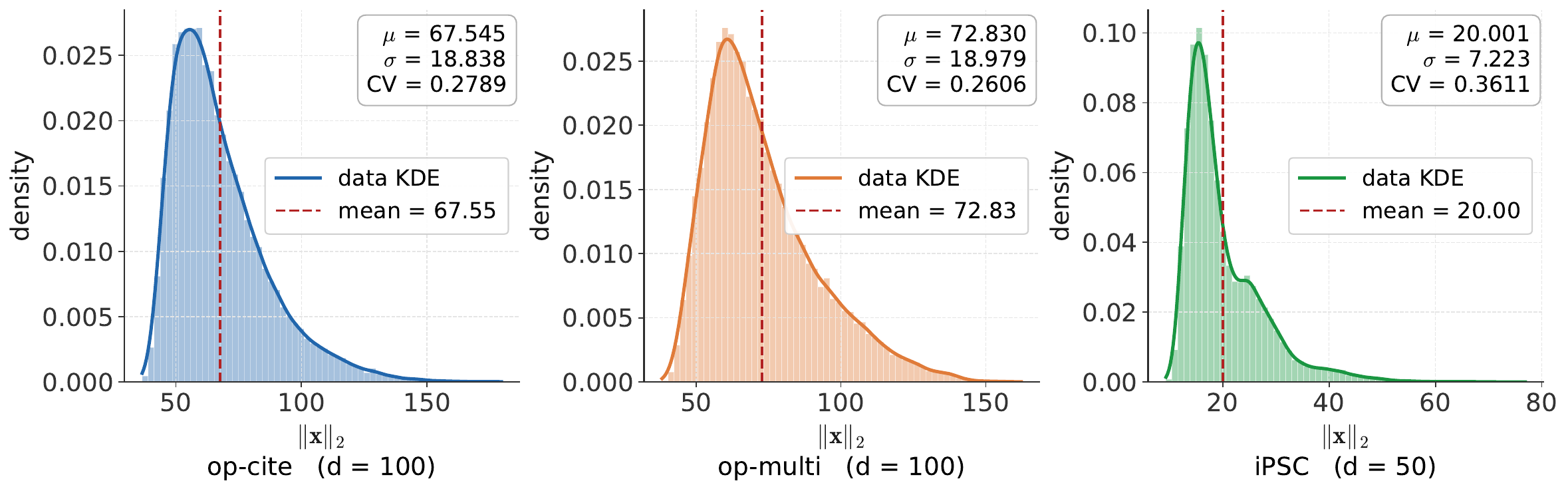}
    \caption{Norm-concentration diagnostics for iPSC-Liu~\citep{liu2020reprogramming} and OP-Cite/OP-Multi~\citep{lance2022multimodal} representations. The panels illustrate the empirical thin-shell behavior summarized by \(\operatorname{CV}_{\rm norm}\), with the \(0.3\) reference level used as a practical concentration warning threshold.}
    \label{fig:app:ipsc_dim}
\end{figure}

\section{iPSC-Liu Results}
\label{sec:app:ipsc}

Figure~\ref{fig:app:ipsc} shows the norm-concentration diagnostics for iPSC-Liu across increasing PCA dimensions, illustrating the $O(d^{-1/2})$ norm-concentration behavior discussed in Appendix~\ref{sec:app:concentration}.
These panels correspond to the high-dimensional benchmark experiments reported in Table~\ref{tab:ipsc_liu_dim10_50} and summarized in the main text in Table~\ref{tab:ipsc_liu_dim10_50_mean}.

\begin{figure}[hbtp]
    \centering
    \includegraphics[width=1.0\linewidth]{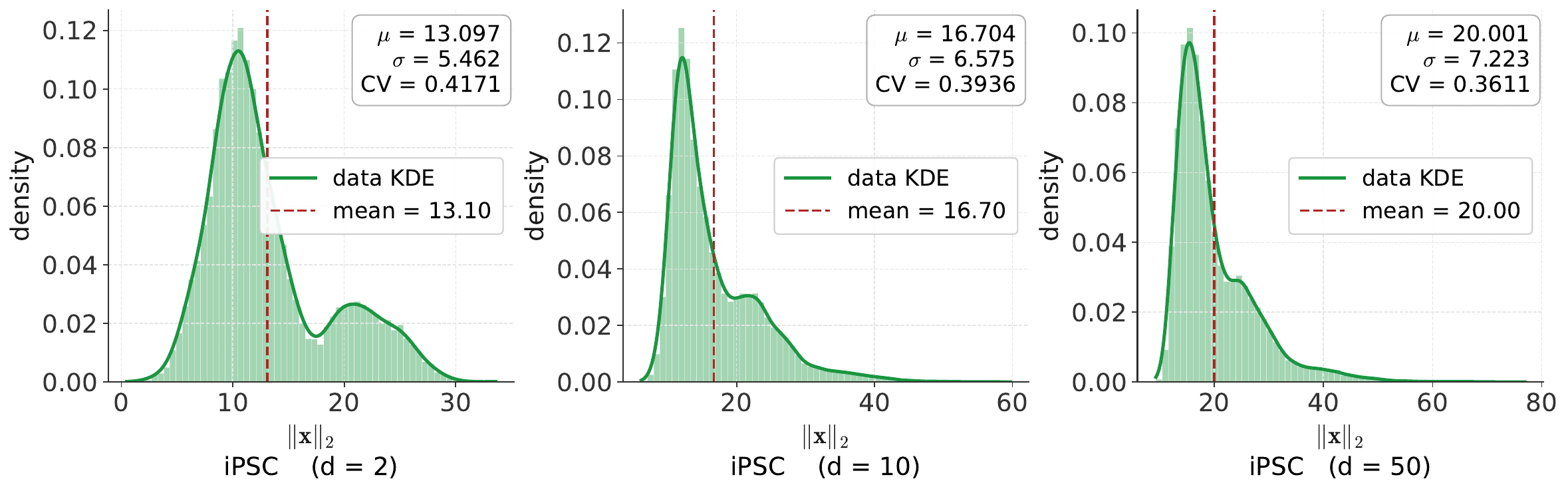}
    \caption{Norm-concentration diagnostics for iPSC-Liu~\citep{liu2020reprogramming} representations. Increasing PCA dimension reduces relative radial variation, consistent with the \(O(d^{-1/2})\) norm-concentration behavior described in Appendix~\ref{sec:app:concentration}.}
    \label{fig:app:ipsc}
\end{figure}

\section{Intuition behind alternating bridge and coupling optimization}
\label{app:alternating-intuition}

\paragraph{Why coupling matters for the bridge.}
The neural bridge is trained on endpoint pairs sampled from the current coupling. If the coupling is biologically implausible (for example, pairing an early stem cell with a terminally differentiated cell), the bridge must learn a path that traverses the entire developmental manifold in a single segment. Such a path is likely to pass through regions where the metric strongly penalizes normal motion, yielding high action and making it difficult for the network to find a low-action deformation. Conversely, when the coupling pairs cells at nearby pseudotime stages, the displacement is small and largely tangent-aligned, so the bridge can easily learn a geometry-consistent interpolant. In short, the coupling determines the training set for the bridge; wrong pairs give the network an impossible learning problem.

\paragraph{Why the bridge matters for the coupling.}
The coupling update solves an OT problem whose cost matrix is the path action evaluated on the neural bridge (Eq.~\eqref{eq:path-cost}). If the bridge is untrained, its paths are nearly straight lines and the cost reduces to Euclidean distance; the OT solver then cannot distinguish a biologically plausible pairing from an implausible one. Once the bridge has learned metric-aware paths, the action cost becomes geometry-sensitive. Pairs whose displacement aligns with the local developmental tangent incur low action, while pairs that would require motion across normal (non-developmental) directions incur high action. The bridge therefore transforms an uninformative Euclidean cost into a geometry-aware cost, allowing OT to select biologically meaningful pairings.

\paragraph{Why amortize with a shared neural network.}
Without amortization, every candidate endpoint pair would require solving an independent boundary-value problem (the geodesic equation) under the state-dependent metric $G_k$. With $N_k N_{k+1}$ candidate pairs, this is computationally infeasible. The neural bridge amortizes this cost by learning a single parametric family $\gamma_\theta(x,y,\tau)$ that approximates the minimum-action path for all pairs. The approximation need not be perfect at initialization; it only needs to be good enough to provide a more informative cost than Euclidean distance, and it improves as the coupling improves.

\section{Adaptive tangent and normal projectors}
\label{app:adaptive-projectors}

The weighted local covariance at each anchor point yields eigenvalues $\lambda_1 \le \lambda_2 \le \dots \le \lambda_d$ and orthonormal eigenvectors $v_1,\dots,v_d$.
The tangent subspace is spanned by the leading (maximum-variance) eigenvectors. Its effective dimension $q_r$ is chosen as the smallest integer such that the cumulative explained variance reaches a prescribed threshold $\tau \in (0,1)$ (e.g., $0.95$):
\[
q_r = \min\!\left\{q \;:\; \frac{\sum_{j=1}^{q} \lambda_{d-j+1}}{\sum_{j=1}^{d} \lambda_j} \ge \tau\right\}.
\]
The tangent projector is $P_T^{(r)} = \sum_{j=1}^{q_r} v_{d-j+1}\,v_{d-j+1}^\top$ and the normal projector is its orthogonal complement $P_N^{(r)} = I - P_T^{(r)}$.
In two dimensions this reduces to $P_N^{(r)} = n_r n_r^\top$ with $n_r = v_1$ and $q_r = 1$.

\section{Adaptive bandwidth estimation}
\label{app:bandwidth}

For each segment $k$, the spatial and temporal bandwidths $h_x^{(k)}$ and $h_t^{(k)}$ are estimated from the anchor normal bank rather than set by hand. The estimation aggregates information over a small window of anchor snapshots centered on segment $k$ (by default one snapshot on each side). The same local eigendecomposition also yields the tangent subspace $P_T^{(r)}$ (Appendix~\ref{app:adaptive-projectors}), which is used for bandwidth estimation below but does not enter the metric tensor directly.

\paragraph{Spatial geometry per snapshot.}
For each anchor snapshot $a$ in the window, PACE computes three quantities from the $k$ nearest neighbors of every cell (excluding the cell itself):

\begin{itemize}
    \item \textbf{Local spatial spacing} $\Delta_x^{(a)}$: the median tangential step size $\|\langle x_j - x_i,\, t_i \rangle\|$ between each cell $i$ and its neighbors $j$, where $t_i$ is the local tangent direction. This captures how densely cells are spaced along the developmental manifold.
    \item \textbf{Normal spatial rate} $\rho_{N,x}^{(a)}$: the median Frobenius-norm difference $\|P_N^{(j)} - P_N^{(i)}\|_F / \Delta_x^{(a)}$ between the normal projectors of neighboring cells. This measures how rapidly the local geometry changes in space.
    \item \textbf{Tangent spatial rate} $\rho_{T,x}^{(a)}$: the analogous quantity for tangent projectors.
\end{itemize}

\paragraph{Cross-snapshot temporal rates.}
For each consecutive pair $(a, a+1)$ in the window, PACE matches cells by nearest-neighbor correspondence and computes:

\begin{itemize}
    \item \textbf{Normal temporal rate} $\rho_{N,t}^{(a,a+1)}$: the Frobenius-norm difference $\|P_N^{(\text{nn}(i))} - P_N^{(i)}\|_F / |t_{a+1} - t_a|$ between matched cells, divided by the experimental time gap. This measures how rapidly the local geometry evolves over time.
    \item \textbf{Tangent temporal rate} $\rho_{T,t}^{(a,a+1)}$: the analogous quantity for tangent projectors.
\end{itemize}

\paragraph{Segment bandwidth assembly.}
All per-snapshot and cross-snapshot quantities are concatenated and their positive medians are taken, yielding four segment-level scalars: $\Delta_x^{(k)}$, $\rho_{N,x}^{(k)}$, $\rho_{N,t}^{(k)}$, and $\rho_{T,t}^{(k)}$. The metric bandwidths are then
\begin{align}
    h_x^{(k)} &= \max\!\left(\Delta_x^{(k)},\; \frac{1}{\rho_{N,x}^{(k)}}\right), \\
    c_N^{(k)} &= \frac{\rho_{N,t}^{(k)}}{\rho_{N,x}^{(k)}}, \\
    h_t^{(k)} &= \frac{h_x^{(k)}}{c_N^{(k)}}.
\end{align}

The spatial bandwidth $h_x$ is the larger of the local spacing and the reciprocal spatial variation scale, preventing the kernel from being narrower than either the point density or the geometry variation. The ratio $c_N^{(k)}$ converts spatial scale into equivalent temporal scale through the speed at which the normal geometry changes in time relative to space. A similar pair $(\sigma_x^{(k)}, \sigma_t^{(k)})$ is computed from the tangent rates for the coupling-refinement kernel, but the metric field uses $(h_x^{(k)}, h_t^{(k)})$.

\section{Limitations and Future Work}
\label{app:limitations}

\paragraph{Local geometry is an estimated prior.}
PACE depends on nearest-neighbor covariance estimates within each snapshot. Sparse, noisy, or heterogeneous observations can make tangent and normal subspaces unstable. Moreover, a covariance subspace describes locally supported, undirected variation; it does not by itself identify the direction of time or the correct fate assignment near a branch. The geometry is therefore an inductive bias for selecting plausible dynamics, rather than a guarantee of biological identifiability.

\paragraph{Extrinsic interpolation and approximate minimum-action paths.}
PACE averages anchor projectors directly in the observed coordinates, using Euclidean spatial distances and ordinary temporal distances. The resulting $C_N^{(k)}(x,t)$ is a soft normal-penalty tensor and is generally not an idempotent projector. Although $0\preceq C_N^{(k)}\preceq I$ ensures $I\preceq G_k\preceq(1+\alpha)I$, these spectral bounds do not establish that the estimated subspaces match the underlying dynamics. Under high curvature, folding, or sparse sampling, ambient neighborhoods can mix intrinsically distant regions and blur the local geometry.

The method neither uses metric-induced geodesic neighborhoods and parallel transport to interpolate subspaces nor solves the time-dependent minimum-action boundary-value problem exactly. Instead, a neural bridge approximately minimizes a temporally discretized action under a time-dependent metric. An intrinsic refinement would require an initial metric, computation of its geometric neighborhoods and transport operators, and subsequent updates of the projector field. PACE does not perform this self-consistent refinement. These are limitations of the geometric approximation, while the constructed ambient-space metric remains positive definite.

\paragraph{Optimization, distillation, and long-horizon prediction.}
Stochastic neural training and periodic rematching approximate the ideal block updates in Appendix~\ref{app:optimization-monotonicity}; they do not guarantee global optimality or a decrease at every training step. Distillation into a global velocity field also does not guarantee that every endpoint-conditioned bridge is reproduced exactly. Small local errors can accumulate during integration without resets.

\paragraph{Biological validation and population mass.}
The main-text distributional metrics evaluate held-out marginal reconstruction and do not directly validate individual biological cell histories. Inferred paths should be interpreted as plausible reconstructions, and validation with biological lineage tracing remains future work. The balanced formulation also normalizes snapshots to unit mass and does not explicitly learn proliferation or death. Relaxing the coupling constraints alone would not add a growth or death term to the distilled ODE.

\paragraph{Computational scope and future directions.}
Coupling refinement adds training cost, and the small Ocean timing experiment in Appendix~\ref{app:efficiency} does not establish large-dataset scalability. Future work includes robust geometry estimation, intrinsic metric refinement, uncertainty quantification, and transport with birth-death or growth terms~\citep{schiebinger2019optimal}. Learning representations that preserve local trajectory geometry is also important in high dimensions, where distance concentration can reduce the contrast of coupling costs.

\clearpage
\section{Training Cost and Inference Efficiency}
\label{app:efficiency}

\paragraph{Measurement setting.}
We compare wall-clock training times on the same Ocean split, with training labels $\{0,2,4,6,8\}$ and held-out labels $\{1,3,5,7\}$, using one NVIDIA H800. Times exclude evaluation and figure generation. This timing study is separate from the main-text benchmark experiments performed on an NVIDIA A100 80GB. The measurements describe the tested implementations and training schedules on this small benchmark; they are not a time-to-equal-accuracy comparison or a scaling study. Timing variability across repeated runs is not reported.

\begin{table}[htbp]
\centering
\caption{Wall-clock training time on Ocean using one NVIDIA H800. PACE uses the Sinkhorn configuration.}
\label{tab:app:training-runtime}
\small
\setlength{\tabcolsep}{10pt}
\begin{tabular}{lr}
\toprule
Method & Training time (s) \\
\midrule
OT-CFM & 2.7 \\
Action Matching & 4.9 \\
CURLY & 17.1 \\
MFM & 20.8 \\
PACE (Sinkhorn) & 35.9 \\
DMSB & 37.5 \\
Aligned CFM (ALI-CFM) & 1664.5 \\
\bottomrule
\end{tabular}
\end{table}

\paragraph{Where the training time is spent.}
PACE's Sinkhorn configuration is slower than the tested single-stage and two-stage deterministic baselines, similar in training time to DMSB, and faster than Aligned CFM. Metric construction takes $0.51$ s, approximately $1.4\%$ of the $35.9$ s total. Sinkhorn matching accounts for approximately $19$ s across initialization and five coupling refinements, including $16.3$ s for the refinements alone.

\paragraph{Inference after distillation.}
PACE, OT-CFM, MFM, CURLY, and Aligned CFM use the same approximately 21K-parameter velocity-network architecture and 100 Euler updates (101 grid points) in the controlled Ocean inference comparison. The reported rollout time is approximately $63$ ms for this shared inference setup; this is a workload-specific measurement, not a per-cell latency or a guarantee of identical runtime across methods and batch sizes. PACE evaluates only the distilled velocity field at inference; local metric construction, endpoint-conditioned bridge evaluation, and OT rematching are training-time operations. This common inference result does not cover the distinct Action Matching or DMSB implementations.

\paragraph{Computational scaling.}
For $K$ adjacent snapshot pairs with $N$ cells per snapshot, a dense rematching pass evaluates $N^2$ candidate endpoint pairs per interval. With $M$ path probes, cost $C_\psi$ for evaluating the bridge position and its time derivative, and metric-query cost $C_G(A,d)$ for $A$ anchors in dimension $d$, constructing the path-cost matrices takes
\begin{equation}
\mathcal{O}\!\left(KN^2M\bigl[C_\psi+C_G(A,d)\bigr]\right).
\end{equation}
An explicit weighted sum of dense anchor projectors has $C_G(A,d)=\mathcal{O}(Ad^2)$. At fixed $K,d,M$, when $A$ grows linearly with $N$, the dense metric-query contribution is cubic in $N$. With a fixed number $q$ of candidate targets per source and fixed batch sizes, evaluating selected paths costs $\mathcal{O}(KNqM[C_\psi+C_G])$, making this contribution quadratic in $N$. The implementation still retains a dense Euclidean fallback cost matrix, so candidate-distance computation and cost-matrix storage remain quadratic. The coupling solver is separate: $J$ dense Sinkhorn iterations cost $\mathcal{O}(JN^2)$ per interval. These are costs per rematching pass; multiple refinements repeat them, while geometry preprocessing, bridge fitting, and distillation add further training cost. A rollout with $N$ particles, $S$ integration steps, and velocity-network cost $C_\phi$ per particle costs $\mathcal{O}(NSC_\phi)$, independently of the metric anchor bank.

\end{document}